\documentclass[11pt,a4paper]{article}
\usepackage{epsfig}
\usepackage{graphics}
\usepackage{subfigure}

\usepackage{epstopdf}
\usepackage{hyperref}

\usepackage[english]{babel}
\usepackage{mdwlist}
\usepackage{graphicx}
\usepackage{hyperref}%\hypersetup{backref, colorlinks=true, linkcolor=black,citecolor=black}
\usepackage{appendix}
\usepackage{dsfont}
\usepackage{tikz}
\usepackage{amsmath}
\usepackage{amsbsy}
\usepackage{amsfonts}
\usepackage{amssymb}
\usepackage{amscd}
\usepackage{amsthm}
\usepackage{empheq}

\usepackage{bm}
\usepackage{xspace}

\makeatletter
\def\XS{\xspace}
\DeclareMathAlphabet{\mathb}{OML}{cmm}{b}{it}
\def\sbm#1{\ensuremath{\mathb{#1}}}                % Style gras italique (necessite amsmath)           
\def\sbmm#1{\ensuremath{\boldsymbol{#1}}}          % Style gras italique (necessite amsmath)           
\def\scu#1{\ensuremath{\mathcal{#1\XS}}}           % Style cursif
\def\sbl#1{\ensuremath{\mathbb{#1}}}              % Style blackboard (necessite bbm)

\def\Ab{{\sbm{A}}\XS}  \def\ab{{\sbm{a}}\XS}
\def\Bb{{\sbm{B}}\XS}  
  
\def\Db{{\sbm{D}}\XS}  
  \def\eb{{\sbm{e}}\XS}

\def\Mb{{\sbm{M}}\XS}

  \def\vb{{\sbm{v}}\XS}
  \def\wb{{\sbm{w}}\XS}
  \def\xb{{\sbm{x}}\XS}
  \def\yb{{\sbm{y}}\XS}
\def\Zb{{\sbm{Z}}\XS}  \def\zb{{\sbm{z}}\XS}

\def\Cc{{\scu{C}}\XS}

\def\Fc{{\scu{F}}\XS}

\def\Pc{{\scu{P}}\XS}

\def\Cbb{{\sbl{C}}\XS}  
  
\def\Ebb{{\sbl{E}}\XS}

\def\Nbb{{\sbl{N}}\XS}  
  
\def\Pbb{{\sbl{P}}\XS}  
  
\def\Rbb{{\sbl{R}}\XS}

\def\alphab      {{\sbmm{\alpha}}\XS}

\def\ker{\textrm{Ker}\,}

\def\sgn{{\mathrm{sign}}}							
\def\supp{{\mathrm{supp}}}
\newsavebox{\fminibox}
\newlength{\fminilength}

  \def\+{^\dagger}
\def\nequiv{\not\kern-.05em\equiv}
\def\egal{\kern-.5em=\kern-.5em}        % Moins d'espace autour de "="
\def\propt{\kern-.2em\propto\kern-.2em} % Idem
\def\intdouble{\int\kern-0.3em\int}
\def\inttriple{\int\kern-0.3em\int\kern-0.3em\int}

\def\rond#1{\overset{\kern-0.33em~_\circ}{#1}}
\def\rondit[#1]#2{\overset{\kern#1~_\circ}{#2}}

\newcommand{\Id}{\ensuremath{\mathrm{Id}}}

\def\qed{\ifmmode\hbox{\hfill\sqb}\else{\ifhmode\unskip\fi%
\nobreak\hfil
\penalty50\hskip1em\null\nobreak\hfil$\blacksquare$
\parfillskip=0pt\finalhyphendemerits=0\endgraf}\fi}

\RequirePackage{amsmath}
\RequirePackage{xspace}

\def\XS{\xspace}
\DeclareMathAlphabet{\mathb}{OML}{cmm}{b}{it}
\def\sbm#1{\ensuremath{\mathb{#1}}}                % Style gras italique (necessite amsmath)           
\def\sbmm#1{\ensuremath{\boldsymbol{#1}}}          % Style gras italique (necessite amsmath)           
\def\scu#1{\ensuremath{\mathcal{#1\XS}}}           % Style cursif
\def\Ab{{\sbm{A}}\XS}  \def\ab{{\sbm{a}}\XS}
\def\Bb{{\sbm{B}}\XS}  
  
\def\Db{{\sbm{D}}\XS}  
  \def\eb{{\sbm{e}}\XS}

\def\Mb{{\sbm{M}}\XS}

  \def\vb{{\sbm{v}}\XS}
  \def\wb{{\sbm{w}}\XS}
  \def\xb{{\sbm{x}}\XS}
  \def\yb{{\sbm{y}}\XS}
\def\Zb{{\sbm{Z}}\XS}  \def\zb{{\sbm{z}}\XS}

\def\Cc{{\scu{C}}\XS}

\def\Fc{{\scu{F}}\XS}

\def\Pc{{\scu{P}}\XS}

\def\MAP{^{\kern1pt{\rm MAP}\kern-1pt}}

\usepackage{color,soul}

\def\Cbb{{\sbl{C}}\XS}  
  
\def\Ebb{{\sbl{E}}\XS}

\def\Nbb{{\sbl{N}}\XS}  
  
\def\Pbb{{\sbl{P}}\XS}  
  
\def\Rbb{{\sbl{R}}\XS}

\def\alphab      {{\sbmm{\alpha}}\XS}

\def\Psib        {{\sbmm{\Psi}}\XS}

\def\PM{\kern0pt^{\textrm{{\scriptsize PM}}}\kern0pt}
\def\MMAP{\kern1pt^{\textrm{{\tiny MMAP}}}\kern-1pt} 

\def\rem#1{}                    % rem{bla bla bla}

 \def\btabu{\begin{tabular}}             \def\etabu{\end{tabular}}

\makeatother

\newtheorem{thmchapter}{Theorem}[section]
\newtheorem{definchapter}[thmchapter]{Definition}
\newtheorem{prop}[thmchapter]{Proposition}
\newtheorem{corollary}[thmchapter]{Corollary}
\newtheorem{lemme}[thmchapter]{Lemma}
\newtheorem{remark}[thmchapter]{Remark}

\title{An analysis of blocks sampling strategies in compressed sensing}

\author{J\'er\'emie Bigot$^{(1)}$, Claire Boyer$^{(2,3)}$ and Pierre Weiss$^{(2,3,4,5)}$ \vspace{0.2cm}  \\
\\
$^{(1)}$ DMIA, Institut Sup\'erieur de l'A\'eronautique et de l'Espace, Toulouse, France\\ 
$^{(2)}$ Institut de Math\'ematiques de Toulouse, IMT-UMR5219, Universit\'e de Toulouse, France\\
$^{(3)}$ CNRS, IMT-UMR5219, Toulouse, France \\
$^{(4)}$ Institut des Technologies Avanc\'ees du Vivant, ITAV-USR3505, Toulouse, France\\
$^{(5)}$ CNRS, ITAV-USR3505, Toulouse, France  \\
{\small {jeremie.bigot}@isae.fr}, {\small {claire.boyer}@math.univ-toulouse.fr}, {\small {pierre.armand.weiss}@gmail.com}\\
 \vspace{0.2cm}}
\begin{document}
\maketitle

\begin{abstract}
Compressed sensing is a theory which guarantees the exact recovery of sparse signals from a small number of linear projections. 
The sampling schemes suggested by current compressed sensing theories are often of little practical relevance since they cannot be implemented on real acquisition systems. 
In this paper, we study a new random sampling approach that consists in projecting the signal over blocks of sensing vectors. 
A typical example is the case of blocks made of horizontal lines in the 2D Fourier plane.
We provide theoretical results on the number of blocks that are required for exact sparse signal reconstruction. 
This number depends on two properties named intra and inter-support block coherence. 
We then show through a series of examples including Gaussian measurements, isolated measurements or blocks in time-frequency bases, that the main result is sharp in the sense that the minimum amount of blocks necessary to reconstruct sparse signals cannot be improved up to a multiplicative logarithmic factor.
The proposed results provide a good insight on the possibilities and limits of block compressed sensing in imaging devices such as magnetic resonance imaging, radio-interferometry or ultra-sound imaging.
\end{abstract}

\section{Introduction}

% The fundamental Shannon-Nyquist theorem claims that sampling a signal at least twice faster than its bandwidth is sufficient to exactly reconstruct the initial signal. Nevertheless, the resulting number of measurements needed can be so large that the storage becomes impossible and the acquisition time too long.
Compressive Sensing is a new sampling theory that guarantees accurate recovery of signals from a small number of linear projections using three ingredients listed below:
\begin{itemize}
 \item \textbf{Sparsity:} the signals to reconstruct should be sparse, meaning that they can be represented as a linear combination of a small number of atoms in a well-chosen basis. A vector $\xb \in \Cbb^n$ is said to be $s$-sparse if its number of non-zero entries is equal to $s$.
 \item \textbf{Nonlinear reconstruction:} a key feature ensuring recovery is the use of non linear reconstruction algorithms. For instance, in the seminal papers \cite{donoho2006compressed,candes2006robust}, it is suggested to reconstruct $\xb$ via the following $\ell_1$-minimization problem:
\begin{align}
\label{pbMin1}
\min_{\zb \in \Cbb^n} \left\| \zb \right\|_1 \qquad \text{such that } \qquad \Ab \zb = \yb,
\end{align}
where $\Ab\in \Cbb^{q\times n}$ ($q\leq n$) is a sensing matrix, $\yb = \Ab \xb \in \Cbb^q$ represents the measurements vector, and $\| \zb \|_1 = \sum_{i=1}^n | z_i |$ for all $\zb = \left( z_1 , \hdots , z_n \right) \in \Cbb^n$. 
 \item \textbf{Incoherence of the sensing matrix:} the matrix $\Ab$ should satisfy an incoherence property described later. If $\Ab$ is perfectly incoherent (e.g. random Gaussian measurements or Fourier coefficients drawn uniformly at random) then it can be shown that only $q=O(s\ln(n))$ measurements are sufficient to perfectly reconstruct the $s$-sparse vector $\xb$.
\end{itemize}

The construction of good sensing matrices $\Ab$ is a keystone for the successful application of compressed sensing. The use of matrices with independent random entries has been popularized in the early papers \cite{candes2006stable,candes2008restricted}. Such sensing matrices have limited practical interest since they can hardly be stored on computers or implemented on practical systems. More recently, it has been shown that partial random circulant matrices \cite{foucart2013mathematical,puy2012universal,rauhut2012restricted} may be used in the compressed sensing context. With this structure a matrix-vector product can be efficiently implemented on a computer by convolving the signal $\xb$ with a random pulse and by subsampling the result. This technique can also be implemented on real systems such as magnetic resonance imaging (MRI) or radio-interferometry \cite{puy2012spread}. However this demands to modify the acquisition device physics, which is often uneasy and costly. 
Another way to proceed consists in drawing $q$ sampling locations among $n$ possible ones, see \cite{candes2006robust,rudelson2008sparse}. This setting, which is the most widespread in applications, is a promising avenue to implement compressed sensing strategies on nearly all existing devices. Its efficiency depends on the \textit{incoherence} between the acquisition and sparsity bases \cite{donoho2001uncertainty,candes2007sparsity}. It is successfully used in radio interferometry \cite{wiaux2009compressed}, digital holography \cite{marim2010compressed} or MRI  \cite{lustig2007sparse} where the measurements are Fourier coefficients. 

To the best of our knowledge, all current compressed sensing theories suggest that the measurements should be drawn independently at random. 
This is impossible for most acquisition devices which have specific acquisition constraints. 
A typical example is MRI, where the samples should lie along continuous curves in the Fourier domain (see e.g. \cite{wright1997MRI,lustig2008fast}). 
As a result, most current implementations of compressed sensing do not comply with theory.
For instance, in the seminal work on MRI \cite{lustig2007sparse}, the authors propose to sample parallel lines of the Fourier domain (see Figure \ref{fig:schema} and \ref{fig:lustig}).
\begin{figure}
\begin{center}
\btabu{@{}ccc}
\includegraphics[height=4cm]{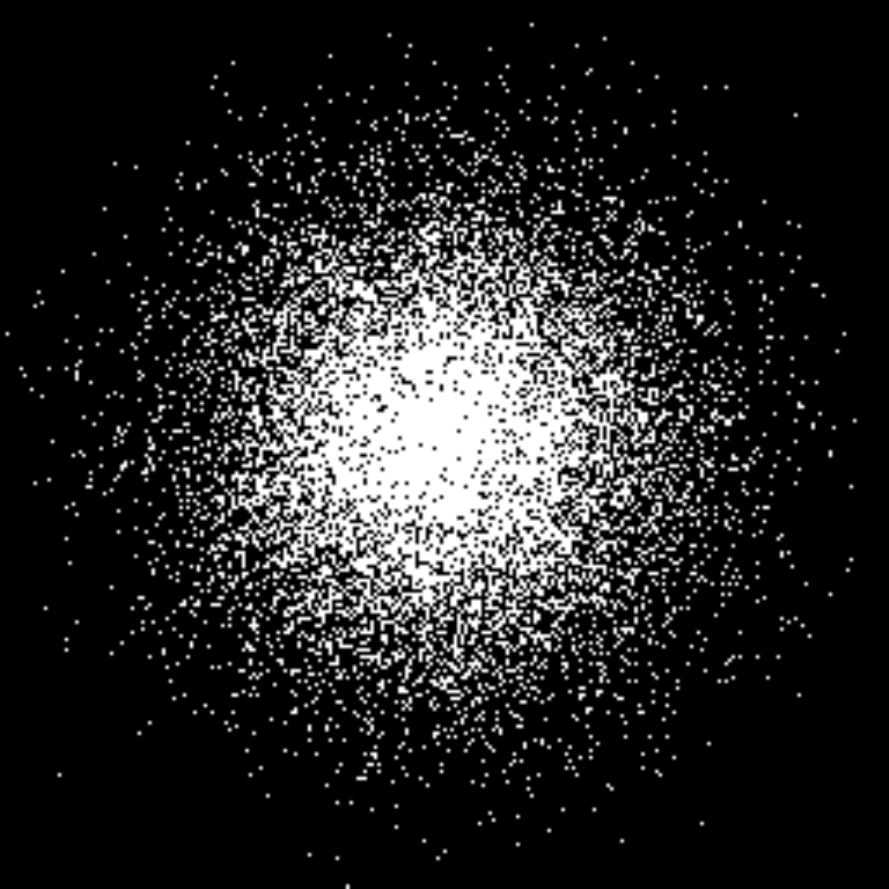} &
\includegraphics[height=4cm]{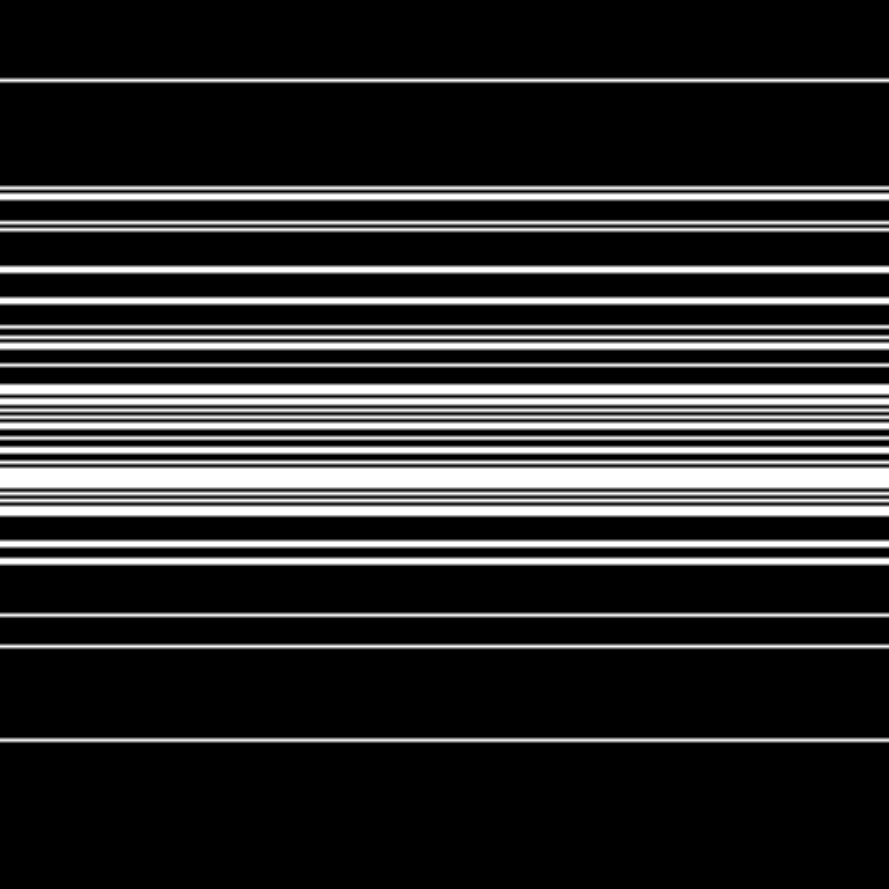} &
\includegraphics[height=4cm]{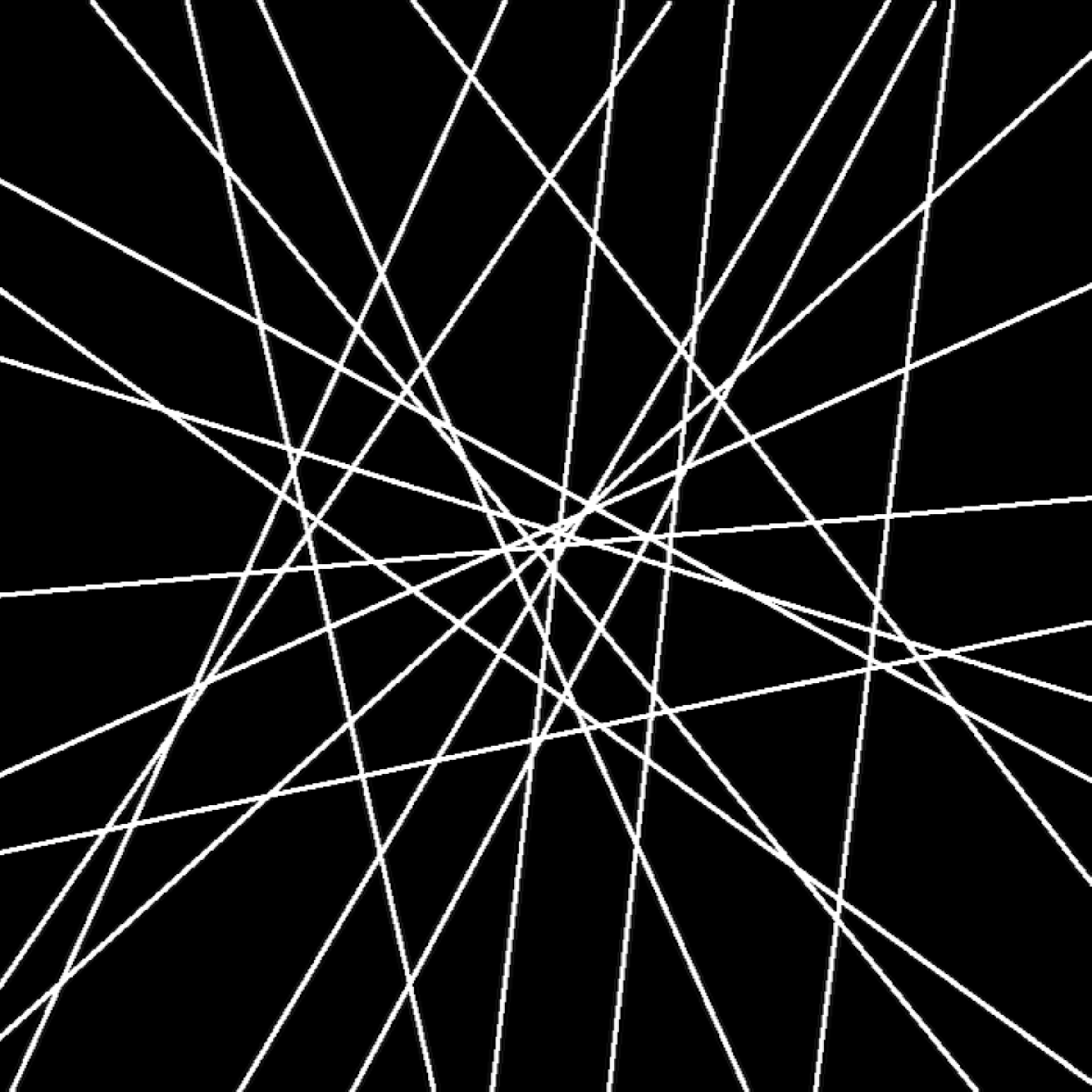} \\
{\small (a)}&{\small (b)} & {\small (c)}
\etabu
\caption{\label{fig:schema}{\bf An example of MRI sampling schemes in the \textit{k-space} (the 2D Fourier plane where low frequencies are centered) } 
(a): Isolated measurements drawn from a probability distribution with radial distribution. 
(b): Sampling scheme in the case of non-overlapping blocks of measurements that correspond to horizontal lines in the 2D Fourier domain. 
(c): Sampling scheme in the case of overlapping blocks of measurements that correspond to straight lines.
}
\end{center}
\end{figure}

\paragraph{Contributions}

In this paper, we aim at bridging the gap between theory and practice. 
We consider a sensing matrix $\Ab$ that is constructed by stacking blocks of measurements and not just isolated measurements. 
In the proposed formalism, the blocks can be nearly arbitrary random matrices. 
For instance, our main result covers the case of blocks made of groups of rows of a deterministic sensing matrix (e.g. lines in the Fourier domain) or blocks with random entries (e.g. Gaussian blocks).
We study the problem of exact non-uniform sparse recovery in a noise-free setting. 
This sampling strategy raises various questions.
How many blocks of measurements are needed to ensure exact reconstruction?
Is the required number of blocks compatible with faster acquisition? 

Our first contribution is to extend the standard compressed sensing theorems to the case of blocks of measurements. 
We then show that our result is sharp in a few practical examples and extends the best currently known results in compressed sensing.
We also prove that in many cases, imposing a block structure has a dramatic effect on the recovery guarantees since it strongly impoverishes the variety of admissible sampling patterns.
Overall, we believe that the presented results give a good theoretical basis to the use of block compressed sensing and show the limits to this setting. This work also provides some insight on many currently used sampling patterns in MRI, echography, computed tomography scanners,...

\paragraph{Related work}

After submitting the first version of this paper, the authors of \cite{polak2012performance} attracted our attention to the fact that their work dealt with a very similar setting. We therefore make a comparison between the results in Section \ref{subsubsec:FourierWavelet}.

\paragraph{Outline of the paper}

The remaining of the paper is organized as follows. In Section \ref{sec:preliminaries}, we first describe the notation and the main assumptions necessary to derive a general theory on blocks of measurements.
We present the main result of this paper about blocks-sampling acquisition in Section \ref{sec:MainResults}.
In Section \ref{sec:tightness}, we discuss the sharpness of our results. First, we show that our approach provides the same guarantees that existing results when using isolated measurements (either Gaussian or randomly extracted from deterministic transforms). 
We conclude on a pathological example to show sharpness in the case of blocks sampled from separable transforms.

\section{Preliminaries}\label{sec:preliminaries}

\subsection{Notation}

Let $S = \left( S_1 , \hdots , S_s\right)$ be a subset of $\{1, \hdots , n\}$ of cardinality $s$.
We denote by $P_S \in \Cbb^{n\times s}$ the matrix with columns $\left( \eb_i \right)_{i\in S}$ where $\eb_i$ denotes the $i$-th vector of the canonical basis of $\Cbb^n$. 
For given $\Mb \in \Cbb^{n \times n}$ and $\vb \in \Cbb^n$, we also define  $\Mb_S = \Mb P_S$, and $\vb_S =P_S^* \vb $. 

\subsection{Main assumptions}

Recall that we consider the following $\ell_1$-minimization problem:
\begin{align}
\label{pb:min}
\min_{\zb \in \Cbb^n} \| \zb \|_1 \qquad \text{s.t.} \quad \yb= \Ab\zb,
\end{align}
where $\Ab$ is the sensing matrix, $\yb =  \Ab \xb \in \Cbb^q$ is the measurements vector, $\xb \in \Cbb^{n}$ is the unknown vector to be recovered. In this paper, we assume that the sensing matrix $\Ab$ can be written as
\begin{align}
\label{eq:samplingMatrixGeneral}
 \Ab = \frac{1}{\sqrt{m}} \begin{pmatrix}
\displaystyle {\Bb_{1}} \\
\vdots \\
\displaystyle {\Bb_{m}} 
\end{pmatrix},
\end{align}
where $\Bb_1, \hdots , \Bb_m$ are i.i.d. copies of a random matrix $\Bb$,  satisfying 
\begin{align}
\label{condIsotropy}
\Ebb\left(  \Bb^* \Bb \right) = \Id,
\end{align} 
where $\Id$ is the $n\times n $ identity matrix. This condition is the extension of the isotropy property described in \cite{candes2011probabilistic} in a blocks-constrained acquisition setting.

In most cases studied in this paper, the random matrix $\Bb$ is assumed to be of fixed size $p \times n$ with $p \in \Nbb^*$.
This assumption is however not necessary. 
The number of blocks of measurements is denoted $m$, while the overall number of measurements is denoted $q$. When $\Bb$ has a fixed size $p\times n$, $q=mp$.

The following quantities will be shown to play a key role to ensure sparse recovery in the sequel.
\begin{definchapter}
\label{def:quantities}
Let $S \subset \{ 1 ,\hdots , n \}$ be a set of cardinality $s$. We denote by
$\left( \mu_i(S) \right)_{1\leq i \leq 3}$ the smallest positive reals such that the following bounds hold either deterministically or stochastically (in a sense discussed later)   
\begin{align}  
\notag
\left\| \Bb_S^* \Bb_S \right\|_{2\rightarrow 2} \leq \mu_1(S), \qquad 
 \sqrt{s}   \max_{i \in S^c} \left\| \Bb_{S}^* \Bb\eb_i   \right\|_{2} \leq \mu_2(S) ,
\\
\label{ineq:conditions}
  s \max_{i \in S^c} \left\| \Ebb \left[ \Bb_{S}^* \left( \Bb \eb_i \right) \left( \Bb \eb_i \right)^* \Bb_S   \right] \right\|_{2\rightarrow 2} \leq \mu_3(S).
\end{align}
Define
$$ \gamma(S) := \max_{1\leq i \leq 3} \mu_i(S).$$
\end{definchapter}

The quantities introduced in Definition \ref{def:quantities} can be interpreted as follows.
The number $\mu_1(S)$ can be seen as an \textit{intra-support} block coherence, whereas $\mu_2(S)$ and $\mu_3(S)$ are related to the \textit{inter-support} block coherence, that is the coherence between blocks restricted to the support of the signal and blocks restricted to the complementary of this support. Note that the factors $\sqrt{s}$ and $s$ involved in the definition of $\mu_2(S)$ and $\mu_3(S)$ ensure homogeneity between all of these quantities.
%In the following sections, we provide some estimates of these quantities for various models of blocks measurements. 

\subsection{Application examples}

The number of applications of the proposed setting is very large. For instance, it encompasses those proposed in \cite{candes2011probabilistic}.
Let us provide a few examples of new applications below.

\subsubsection{Partition of orthogonal transforms}\label{subsubsec:partition}

Let $\Ab_0 \in \Cbb^{n\times n}$ denote an orthogonal transform. 
Blocks can be constructed by partitioning the rows $\left(  \ab_i^* \right)_{1\leq i \leq n}$ from $\Ab_0$:
$$ 
\Bb_j = \left(  \ab_i^* \right)_{i\in I_j} \quad \text{for} \quad I_j \subset \{1, \hdots , n \} \quad \text{s.t.} \quad \bigsqcup_{j=1}^M I_j = \{1, \hdots , n \},
$$
where $\bigsqcup$ stands for the disjoint union. This case is the one studied in \cite{polak2012performance}.

Let $\Pi = \left( \pi_1 , \hdots , \pi_M \right)$ be a discrete probability distribution on the set of integers $\left\lbrace 1, \hdots , M \right\rbrace$.
A random sensing matrix $\Ab$ can be constructed by stacking $m$ i.i.d. copies of the random matrix $\Bb$ defined by $\Pbb(\Bb=\Bb_k/\sqrt{\pi_k})=\pi_k$ for all $k\in \{ 1, \hdots , M \}$.
Note that the normalization by $1/\sqrt{\pi_k}$ ensures that the isotropy condition $\Ebb\left[  \Bb^*\Bb\right]=\Id_n$ is verified.

\subsubsection{Overlapping blocks issued from orthogonal transforms}\label{subsubsec:overlap}

In the last example, we concentrated on partitions, i.e. non-overlapping blocks of measurements. 
The case of overlapping blocks can also be handled. 
To do so, define the blocks $\left( \Bb_j \right)_{1\leq j \leq M}$ as follows:
$ \Bb_j = \left( \frac{1}{\sqrt{\alpha_i}} \ab_i^* \right)_{i \in I_j}, 
$
where $\displaystyle \bigcup_{j=1}^M I_j = \{1, \hdots  , n\}$, and $\alpha_i$ denotes the multiplicity of the row $\ab_i^*$, i.e. the number of appearances $\alpha_i = |\{j, i \in I_j \}|$ of this row in different blocks.
This renormalization is sufficient to ensure $\Ebb \left[ \Bb^* \Bb \right] = \Id_n$ where $\Bb_k$ is defined similarly to the previous example. 
See Appendix \ref{app:overlap} for an illustration of this setting in the case of 2D Fourier measurements.

\subsubsection{Blocks issued from tight or continuous frames}\label{subsubsec:frames}

Until now, we have concentrated on projections over a fixed set of $n$ vectors. 
This is not necessary and the projection set can be redundant and even infinite. 
A typical example is the Fourier transform with a continuous frequency spectrum. 
This example is discussed in more details in \cite{foucart2013mathematical,candes2011probabilistic}.

\subsubsection{Random blocks}\label{subsubsec:gaussian}

In the previous examples, the blocks were predefined and extracted from deterministic matrices or systems.
The proposed theory also applies to random blocks.
For instance, one could consider blocks with i.i.d. Gaussian entries since these blocks satisfy the isotropy condition \eqref{condIsotropy}. 
This example is of little practical relevance since stacking random Gaussian matrices produces a random Gaussian matrix that can be analyzed with standard compressed sensing approaches. It however presents a theoretical interest in order to show the sharpness of our main result.
Another example with potential interest is that of blocks generated randomly using random walks over the acquisition space \cite{chauffert2014}.

\section{Main result\label{sec:MainResults}}

Our main result reads as follows.
\begin{thmchapter}
\label{thm:recovery}
Let $S \subset \{ 1 ,\hdots , n \}$ be a set of indices of cardinality $s$  and suppose that $\xb \in \Cbb^{n}$ is an $s$-sparse vector supported on $S$.  Fix $\varepsilon \in (0,1)$. Suppose that the sampling matrix $\Ab$ is constructed as in \eqref{eq:samplingMatrixGeneral}, and that the isotropy condition \eqref{condIsotropy} holds.  Suppose that the bounds \eqref{ineq:conditions} hold deterministically. If the number of blocks $m$ satisfies the following inequality
\begin{eqnarray*}
m &\geq & c {\gamma(S)} \log \left( 4n \right) \log\left (12 \varepsilon^{-1}\right),
% \left( 2 \log \left( 4n \right) \log\left (12 \varepsilon^{-1}\right) + \log s \log \left( 12e \log(s) \varepsilon^{-1}\right)\right)
\end{eqnarray*}
then $\xb$ is the unique solution of \eqref{pb:min}  with probability at least $1-\varepsilon$. The constant $c$ can be taken equal to $ 3 \times 534$.
\end{thmchapter}
The proof of Theorem \ref{thm:recovery} is detailed in Section \ref{sec:proof1}.  It is based on the so-called golfing scheme introduced in \cite{gross2011recovering} for matrix completion, and adapted by \cite{candes2011probabilistic} for compressed sensing from isolated measurements. Note that Theorem \ref{thm:recovery} is a non uniform result in the sense that reconstruction holds for a given support $S$ and not for all $s$-sparse signals. It is likely that uniform results could be derived by using the so-called Restricted Isometry Property. However, this strong property is usually harder to prove and leads to narrower classes of admissible matrices and to larger number of required measurements. 

\begin{remark}[The case of stochastic bounds] In Definition \ref{def:quantities}, we say that the bounds \emph{deterministically hold} if the inequalities \eqref{ineq:conditions} are satisfied almost surely. This assumption is convenient to simplify the proof of Theorem \ref{thm:recovery}. 
Obviously, it is not satisfied in the setting where the entries of $\Bb$ are i.i.d. Gaussian variables. 
To encompass such cases, the bounds in Definition \ref{def:quantities} could \emph{stochastically hold}, meaning that the inequalities \eqref{ineq:conditions} are satisfied with large probability. 
The proof of the main result can be modified by conditioning the deviation inequalities in the Lemmas of Appendix \ref{sec:proof1} to the event that the bounds in Definition \ref{def:quantities} hold. 
%Unfortunately, the conditional distribution in the estimates does not verify the isotropy condition \eqref{condIsotropy}. 
%To overcome this issue, a near-isotropy condition as in  \cite[Appendix B]{candes2011probabilistic} can be used. 
Therefore, even though we do not provide a detailed proof, the lower bound on the required number of blocks in Theorem \ref{thm:recovery} remains accurate. Hence, we will propose in Section \ref{sec:Gauss} some estimates of the quantities \eqref{ineq:conditions} in the case of Gaussian measurements.
\end{remark}

The lower bound on the number $m$ of blocks of measurements in Theorem \ref{thm:recovery} depends on $\gamma(S)$ and thus on the support $S$ of the vector $\xb$ to reconstruct. 
In the usual compressed sensing framework, the matrix $\Ab$ is constructed by stacking realizations of a random vector $\ab$.  
The best known results state that $O(s\mu \log(n))$ isolated measurements are sufficient to reconstruct $\xb$ with high probability. 
The coherence $\mu$ is the smallest number such that $\|\ab\|_\infty^2\leq \mu$.
The quantity $\gamma(S)$ in Theorem \ref{thm:recovery} therefore replaces the standard factor $s\mu$. 
The coherence $\mu$ is usually much simpler to evaluate than $\gamma(S)$ which depends on three properties of the random matrix $\Bb$: the intra-support coherence $\mu_1$ and the inter-support coherences $\mu_2$ and $\mu_3$. 
As will be seen in Section \ref{sec:tightness}, it is important to keep all those quantities in order to obtain tight reconstruction results. Nevertheless, a rough upper bound of $\gamma(S)$, reminiscent of the coherence, can be used as shown in Proposition \ref{prop:gammaSs}.
\begin{prop}
\label{prop:gammaSs}
Let $S$ be a subset of $\{ 1, \hdots , n \}$ of cardinality $s$. 
Assume that  the following inequality holds either deterministically or stochastically
$$
\|\Bb ^* \Bb\|_{1\rightarrow \infty}~\leq~\mu_4
$$
with $\displaystyle \| \Bb ^* \Bb \|_{1\rightarrow \infty} = \sup_{\| \vb \|_1 \leq 1} \| \Bb ^* \Bb \vb \|_{\infty}$.
Then
\begin{align} 
\label{gamma}
\gamma(S) \leq s \mu_4.
\end{align}
\end{prop}
The proof of Proposition \ref{prop:gammaSs} is given in Appendix \ref{proof:gamma}.
The bound given in Proposition \ref{prop:gammaSs} is an upper bound on $\gamma(S)$ that should not be considered as optimal. For instance, for Gaussian measurements, it is important to precisely evaluate the three quantities $(\mu_i(S))_{1\leq i \leq 3}$. 
%It will be shown in Section \ref{sec:tightness} that this upper bound is sharp in a few settings of interest.

\begin{remark}[Noisy setting] 
In this paper, we concentrate on a noiseless setting. It is likely that noise can be accounted for mimicking the proofs in \cite{candes2011probabilistic} for instance.
\end{remark}

\section{Sharpness of the main result}
\label{sec:tightness}

%We recall that the overall number of measurements is $q=mp$, where $m$ is the number of sampled blocks. 
In this section, we discuss the sharpness of the lower bound given by Theorem \ref{thm:recovery} by comparing it to the best known results in compressed sensing.

\subsection{The case of isolated measurements}\label{subsec:isolmeas}

First, let us show that our result matches the standard setting where the blocks are made of only one row, that is $p=1$.
This is the standard compressed sensing framework considered e.g. by \cite{candes2006robust,foucart2013mathematical,candes2011probabilistic}.  Consider that $\Ab_0 = \left( \ab_i^* \right)_{1\leq i \leq n}$ is a deterministic matrix, and that the sensing matrix $\Ab$ is constructed by drawing $m$ rows of $\Ab_0$ according to some probability distribution $\Pc = \left( p_1, \hdots , p_n\right)$, i.e. one can write $\Ab$ as follows:
$$ \Ab  = \begin{pmatrix}
\frac{\ab_{J_1}^*}{\sqrt{p_{J_1}}} \\ \vdots \\ \frac{\ab_{J_m}^*}{\sqrt{p_{J_m}}}
\end{pmatrix},
$$
where the $\left( J_j \right)_{1\leq j \leq m}$'s are i.i.d. random variables taking their value in $\{1, \hdots , n\}$ with probability $\Pc$. According to Proposition \ref{prop:gammaSs}, for a support $S$ of cardinality $s$ the following upper bound holds:
$$ \gamma(S) \leq s \max_{1 \leq j \leq M} \frac{\|\ab_j \ab_j^* \|_{1\rightarrow \infty}}{p_j}.
$$
Therefore, according to Theorem \ref{thm:recovery}, it is sufficient that 
\begin{align}  
\label{ineq:oneRowResult1}
 q \geq c s \max_{1 \leq j \leq M} \frac{\|\ab_j \ab_j^* \|_{1\rightarrow \infty}}{p_j} \log \left( 4n \right) \log\left (12 \varepsilon^{-1}\right). 
\end{align}
to obtain perfect reconstruction with probability $1-\varepsilon$.
Noting that $\|\ab_j \|_\infty^2 = \|\ab_j \ab_j^* \|_{1\rightarrow \infty}$ ,  for all $ j \in \left\lbrace 1, \hdots , n \right\rbrace$,
it follows that Condition \eqref{ineq:oneRowResult1} is the same (up to a multiplicative constant) to that of \cite{candes2011probabilistic}. 
%This difference is not too prejudicial since Theorem \ref{thm:recovery} should be mainly considered as a guide to construct sampling schemes and not as a requirement for perfect recovery. 

In addition, choosing $\Pc^\star$ in order to minimize the right-hand side of \eqref{ineq:oneRowResult1} leads to 
$$p^\star_j= \frac{\|\ab_j \ab_j^* \|_{1\rightarrow \infty}}{\sum_{k=1}^n \|\ab_k \ab_k^*\|_{1\rightarrow \infty}}, \qquad \forall k \in \left\lbrace 1 , \hdots , n \right\rbrace,
$$
which in turn leads to the following required number of measurements:
\begin{align}  
 q &\geq 
 c s \sum_{k=1}^n \|\ab_k^*\|_\infty^2 \log \left( 4n \right) \log\left (12 \varepsilon^{-1}\right).
\end{align}
Contrarily to common belief, the probability distribution minimizing the required number of measurements is not the uniform one, but the one depending on the $\ell_\infty$-norm of the considered row. 
% Also note that $\sum_{\ell=1}^n \|\ab_\ell \|_\infty^2$ does not correspond to the traditional definition of coherence.
Let us highlight this fact. Consider that $\Ab_0 = \begin{pmatrix}
1 & 0 \\
0 & \Fc_{n-1}
\end{pmatrix}$, where $\Fc_{n-1}$ denotes the 1D Fourier matrix of size $(n-1)\times (n-1)$. If a uniform drawing distribution is chosen, the right hand side of \eqref{ineq:oneRowResult1} is $O(sn\ln^2(n))$. This shows that uniform random sampling is not interesting for this sensing matrix. 
Note that the coherence $\|\Ab_0\|_{1\rightarrow\infty}^2$ of $\Ab_0$ is equal to $1$, which is the worst possible case for orthogonal matrices. Nevertheless, if the optimal drawing distribution is chosen, i.e.\
$$ p_j^\star = \left\{
\begin{array}{ll}
\frac{1}{2} & \text{if} \; j =1 \\
\frac{1}{2(n-1)} & \text{otherwise}
\end{array}
\right.
$$ 
then, the right hand side of \eqref{ineq:oneRowResult1} becomes $O(2s\ln^2(n))$. 
Using this sampling strategy, compressed sensing therefore remains relevant. 
Furthermore, note that the latter bound could be easily reduced by a factor 2 by systematically sampling the location associated to the first row of $\Ab_0$, and uniformly picking the $q-1$ remaining isolated measurements. 
Similar remarks were formulated in \cite{krahmer2013stable} which promote non-uniform sampling strategies in compressed sensing.

%\subsection{The case of blocks constructed by randomly drawn isolated measurements}

\subsection{The case of Gaussian measurements \label{sec:Gauss}}

We suppose that the entries of $\Bb \in \Rbb^{p\times n}$ are i.i.d. Gaussian random variables with zero-mean and variance $1/p$. This assumption on the variance ensures that the isotropy condition \eqref{condIsotropy} is satisfied.
The matrix $\Ab$ constructed by concatenating such blocks is also a Gaussian random matrix with i.i.d. entries and does not differ from an acquisition setting based on isolated measurements. 
Therefore, if Theorem \ref{thm:recovery} is sharp, one can expect that $q=O(s\log(n))$ measurements are enough to perfectly reconstruct $\xb$. In what follows, we show that this is indeed the case.
\begin{prop}
\label{prop:gaussian}
Assume that the entries of $\Bb \in \Rbb^{p\times n}$ are i.i.d. Gaussian random variables with zero-mean and variance $1/p$. Then, $\gamma(S) = O\left( \frac{s \log(s) }{p}\right)$. Therefore, $O\left( \frac{s \log(s) \log(n)}{p} \right)$ Gaussian blocks are sufficient to ensure perfect reconstruction with high probability.
\end{prop}
This is similar to an acquisition based on isolated Gaussian measurements and this is optimal up to a logarithmic factor, see \cite{donoho2006compressed}.

%This suggests that Theorem \ref{thm:recovery} is tight in the sense that we preserve the minimal amount of isolated measurements, when the acquisition is done by blocks of measurements.

\subsection{The case of separable transforms \label{sec:tightnessSep}}

In this section, we consider $d$-dimensional deterministic transforms obtained as Kronecker products of orthogonal one-dimensional transforms. 
This setting is widespread in applications. 
Indeed, separable transforms include $d$-dimensional Fourier transforms met in astronomy \cite{bobin2008compressed} or products of Fourier and wavelet transforms met in MRI \cite{lustig2008compressed} or radio-interferometry \cite{wiaux2009compressed}.
A specific scenario encountered in many settings is that of blocks made of lines in the acquisition space.  
For instance, parallel lines in the 3D Fourier space are used in \cite{lustig2007sparse}. The authors propose to undersample the 2D $k_x$-$k_y$ plane and sample continuously along the orthogonal direction $k_z$ (see Figure \ref{fig:lustig}).
\begin{figure}
\btabu{@{}cc}
\includegraphics[height=6cm]{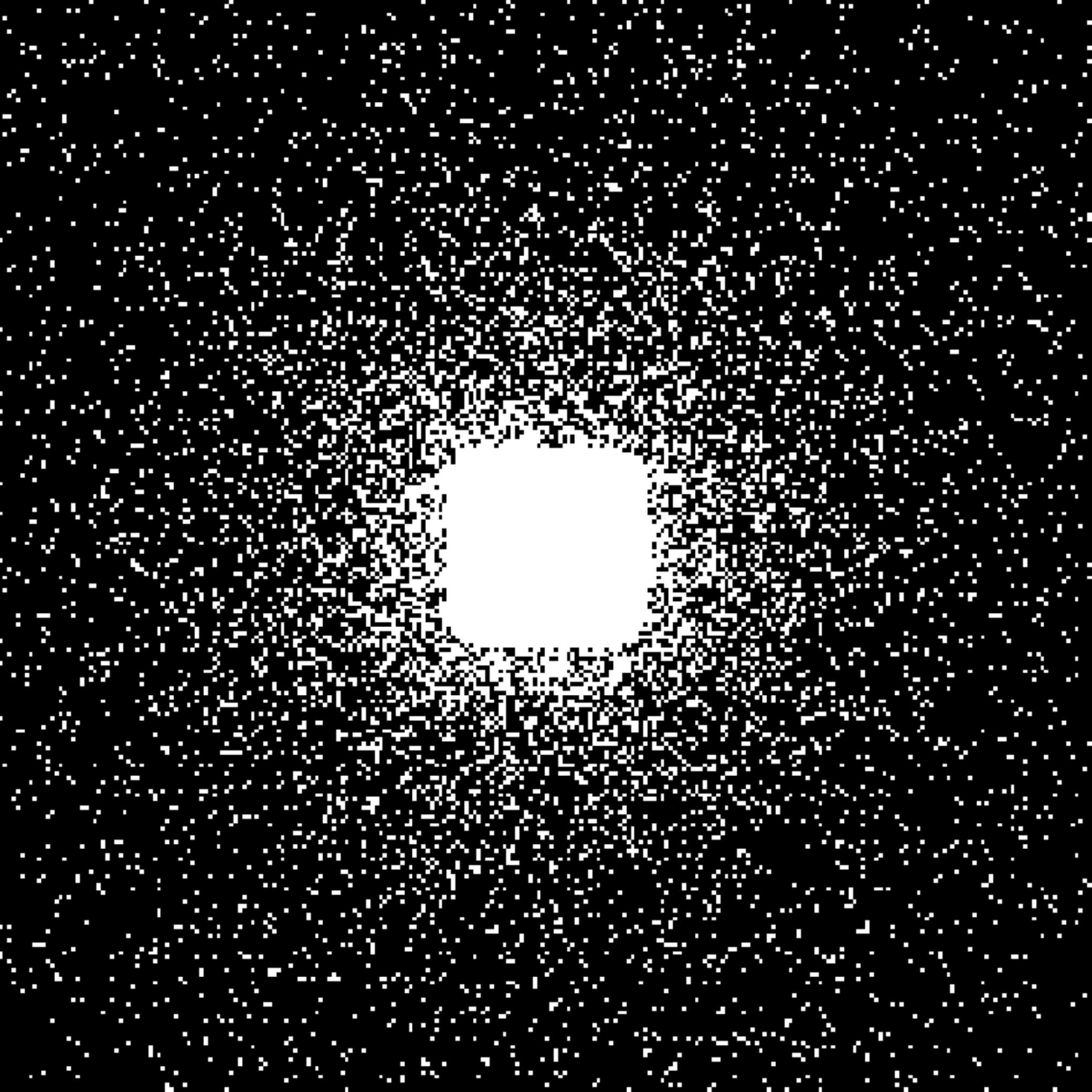} &
\includegraphics[height=6cm]{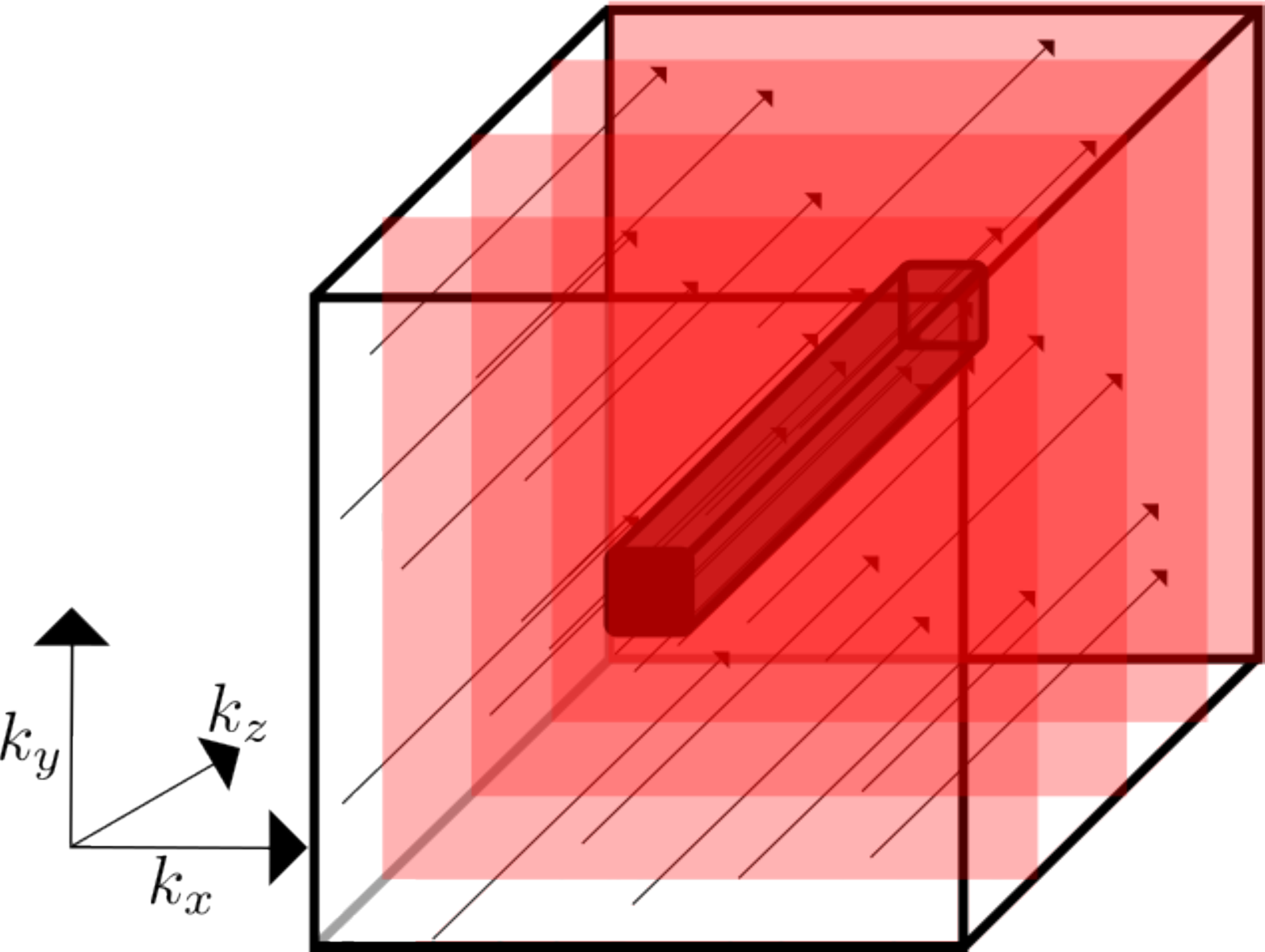} \\
{\small (a)}&{\small (b)}
\etabu
\caption{\label{fig:lustig} Example of sampling pattern used in MRI \cite{lustig2007sparse}. (a) Visualization in the $k_x$-$k_y$ plane. (b) Visualization in 3D.}
\end{figure}

The remaining of this Section is as follows.
We first introduce the notation. 
We then provide theoretical results about the minimal amount of blocks necessary to reconstruct all $s$-sparse vectors. 
Next, we show that Theorem \ref{thm:recovery} is sharp in this setting since the amount of blocks required to reconstruct $s$-sparse vectors coincides with the minimal amount. Finally, we perform a comparison with the results in \cite{polak2012performance}.

\subsubsection{Preliminaries}
\label{subsubsec:formalism}

Let $\Psib \in \Cbb^{\sqrt{n}\times \sqrt{n}}$ denote an arbitrary orthogonal transform, with $\sqrt{n} \in \Nbb$.
Let 
\begin{equation*}
 \Ab_0= \Psib \otimes \Psib = \begin{bmatrix}
                        \Psib_{1,1} \Psib & \hdots & \Psib_{1,\sqrt{n}} \Psib \\ 
		        \Psib_{2,1} \Psib & \hdots & \Psib_{2,\sqrt{n}} \Psib \\ 
			\vdots & \ddots & \vdots \\
		        \Psib_{\sqrt{n},1} \Psib & \hdots & \Psib_{\sqrt{n},\sqrt{n}} \Psib \\ 	
                       \end{bmatrix} \in \Cbb^{n\times n},
\end{equation*}
where $\otimes$ denote the Kronecker product. 
Note that $\Ab_0$ is also orthogonal.
We define groups of measurements from $\Ab_0$ as follows:
\begin{align}
 \Bb_k &= \Psib_{k,:} \otimes \Psib \label{eq:blocktensor}\\
     & = \begin{bmatrix} \Psib_{k,1} \Psib, \hdots, \Psib_{k,\sqrt{n}} \Psib\end{bmatrix} \in \Cbb^{\sqrt{n}\times n}.
\end{align}
For instance, if $\Psib$ is the 1D discrete Fourier transform, this strategy consists in constructing $\sqrt{n}$ blocks as horizontal discrete lines of the discrete Fourier plane. This is similar to the blocks used in \cite{lustig2007sparse}.
Similarly to paragraph \ref{subsubsec:partition}, a sensing matrix $\Ab$ can be constructed by drawing $m$ i.i.d. blocks with distribution $\Pi$. 
Letting
$ K = (k_1 , \hdots , k_m )
\in \{1,\hdots,\sqrt{n}\}^m
$
denote the drawn blocks indexes, $\Ab$ reads:
\begin{align}
 \Ab&= \frac{1}{\sqrt{m}}\begin{bmatrix}
     \frac{\Bb_{k_1}}{\sqrt{\pi_{k_1}}} \\ \vdots \\ \frac{\Bb_{k_m}}{\sqrt{\pi_{k_m}}}
    \end{bmatrix} \label{sampMatrixspe} \\
&= \left(\frac{\Db(\pi)^{-1/2}}{\sqrt{m}} \cdot \Psib_{K,:}\right)  \otimes \Psib \nonumber \\
&= \widetilde{\Psib}_{K,:} \otimes \Psib \nonumber \\
\end{align}
where $\Db(\pi):=\mathrm{diag}(\pi_{k_1},\hdots,\pi_{k_m})$ and $\widetilde{\Psib}_{K,:} :=\frac{\Db(\pi)^{-1/2}}{\sqrt{m}} \cdot \Psib_{K,:}$.
By combining the results in Theorem \ref{thm:recovery} and Proposition \ref{prop:gammaSs}, we easily get the following reconstruction guarantees.
\begin{prop}
\label{prop:optimalPi}
Let $S \subset \{1, \hdots , n \}$ be the support of cardinality $s$ of the signal $\xb \in \Cbb^n$ to reconstruct. 
Under the above hypotheses, if 
\begin{align} 
\label{eq:nbmesnec1}
m \geq c s \max_{1 \leq j \leq M} \frac{\|\Bb_j^* \Bb_j \|_{1\rightarrow \infty}}{\pi_j}  \log \left( 4n \right) \log\left (12 \varepsilon^{-1}\right),
\end{align}
then the vector $\xb$ is the unique solution of \eqref{pb:min}  with probability at least $1-\varepsilon$.
\end{prop}

Using the above result we also obtain the following Corollary.
\begin{corollary}\label{cor:optimalpi}
The drawing probability distribution $\Pi^\star$ minimizing the right hand side of Inequality \eqref{eq:nbmesnec1} on the required number of measurements is defined by
\begin{align}
\pi_j^\star &= \frac{\left\| \Bb_j^* \Bb_j \right\|_{1\rightarrow \infty}}{\sum_{k=1}^M \left\| \Bb_k^* \Bb_k \right\|_{1\rightarrow \infty}} , \qquad \forall j \in \left\lbrace 1 , \hdots , M , \right\rbrace.
\end{align}
For this particular choice of $\Pi^\star$, the right hand side of Inequality \eqref{eq:nbmesnec1} can be written as follows
\begin{align}
m \geq c s  \sum_{j=1}^M \|\Bb_j^* \Bb_j \|_{1\rightarrow \infty} \log \left( 4n \right) \log\left (12 \varepsilon^{-1}\right) .
\end{align}
\end{corollary}

The sharpness of the bounds on the required number of measurements in Proposition \ref{cor:optimalpi} will be discussed in the following paragraph. 

\subsubsection{The limits of separable transforms}\label{subsubsec:limits}

Considering a 2D discrete Fourier transform and a dictionary of blocks made of horizontal lines in the discrete Fourier domain, one could hope to only require $m=O(s/p\log(n))$ blocks of measurements to perfectly recover all $s$-sparse vectors. Indeed, it is known since \cite{candes2006robust} that $O(s\log(n))$ isolated measurements uniformly drawn at random are sufficient to achieve this. 
In this paragraph, we show that this expectation cannot be satisfied since at least $2s$ blocks are necessary to reconstruct all $s$-sparse vectors.
It means that this specific block structure is inadequate to obtain strong reconstruction guarantees.
This result also shows that Theorem \ref{cor:optimalpi} is nearly optimal.

In order to prove those results, we first recall the following useful lemma.
We define a decoder as any mapping $\Delta : \Cbb^q \rightarrow \Cbb^n$. Note that $\Delta$ is not necessarily a linear mapping. 
\begin{lemme}\cite[Lemma 3.1]{cohen2009compressed}
\label{lem:cohen}
Set $\Sigma_s$ to be the set of $s$-sparse vectors in $\Cbb^n$.
If $\Ab$ is any $m\times n$ matrix, then the following propositions are equivalent:
\begin{enumerate}
\item There is a decoder $\Delta$ such that $\Delta(\Ab \xb ) = \xb$, for all  $s$-sparse $\xb$ in $\Cbb^n$.
\item $\Sigma_{2s} \cap \ker \Ab = \{0\}$.
\item For any set $T\subset \{ 1 ,\hdots , n \}$ of cardinality $2s$, the matrix $\Ab_T$ has rank $2s$.
\end{enumerate}
\end{lemme}
Looking at (iii)  of Lemma \ref{lem:cohen}, since the rank of ${ \Ab_T} $ is smaller than $\min(2s,m)$, we deduce that $m\geq 2s$ is a necessary condition to have a decoder. Therefore, if the number of isolated measurements is less than $2s$ with $s$ the degree of sparsity of $\xb$, we cannot reconstruct $\xb$. This property is an important step to prove Proposition \ref{prop:tensorlimit}.

\begin{prop}\label{prop:tensorlimit}
Assume that the sensing matrix $\Ab$ has the special block structure described in \eqref{sampMatrixspe}. If $m< \min(2s,\sqrt{n})$, then there exists no decoder $\Delta$ such that $ \Delta(\Ab \xb ) = \xb$ for all $s$-sparse vector $\xb\in \Cbb^n$.
%a class of $s$-sparse vectors $\Cc_s$ such that for every $\xb $,  $\Delta(\Ab \xb ) \neq \xb$ for any decoder $\Delta$.
In other words, the minimal number $m$ of distinct blocks required to identify every $s$-sparse vectors is necessarily larger than $\min(2s,\sqrt{n})$. 
\end{prop}
Proposition \ref{prop:tensorlimit} shows that there is no hope to reconstruct all $s$-sparse vectors with less than $m=O(s)$ blocks of measurements, using sensing matrices $\Ab$ made of blocks such as \eqref{eq:blocktensor}. Moreover, since the blocks are of length $p=\sqrt{n}$, it follows that whenever $s\geq \frac{\sqrt{n}}{2}$, the full matrix $\Ab_0$ should be used to identify every $s$-sparse $\xb$.
Let us illustrate this result on a practical example. 
Set $\Ab_0$ to be the 2D Fourier matrix, i.e. the Kronecker product of two 1D Fourier matrices. Consider that the dictionary of blocks is made of horizontal lines. Now consider a vector $\xb \in \Rbb^{32\times 32}$ to be $10$-sparse in the spatial domain and only supported on the first column as illustrated in Figure \ref{fig:pathoSignal}(a). 
Due to this specific signal structure, the Fourier coefficients of $\xb$ are constant along horizontal lines, see Figure \ref{fig:pathoSignal}(b). 
Therefore, for this type of signal, the information captured by a block of measurements (i.e. a horizontal line) is as informative as one isolated measurement. 
Clearly, at least $O(s)$ blocks are therefore required to reconstruct all $s$-sparse vectors supported on a vertical line of the 2D Fourier plane.
Using Corollary \ref{cor:optimalpi}, one can derive the following result.
\begin{prop}
\label{prop:2DFour2}
Let $\Ab_0 \in \Cbb^{n \times n}$ denote the 2D discrete Fourier matrix and consider a partition in $M=\sqrt{n}$ blocks that consist of lines in the 2D Fourier domain.
Assume that  $\xb \in \Cbb^n$ is $s$-sparse.
The drawing probability minimizing the right hand side of \eqref{eq:nbmesnec1} is given by 
$$ 
\pi_j^\star = \frac{1}{\sqrt{n}}, \quad \forall j \in  \left\lbrace 1,\hdots,\sqrt{n} \right\rbrace
$$
and for this particular choice, the number $m$ of blocks of measurements sufficient to reconstruct $\xb$ with probability $1-\varepsilon$ is
$$ m \geq c s  \log \left( 4n \right) \log\left (12 \varepsilon^{-1}\right).
$$
\end{prop}
This result is disappointing but optimal up to a logarithmic factor, due to Proposition \ref{prop:tensorlimit}.
We refer to Appendix \ref{comput2DFourier} for the proof. This Proposition indicates that $O(s\log(n))$ blocks are sufficient to reconstruct $\xb$ which is similar to the minimal number given in Proposition \ref{prop:tensorlimit} up to a logarithmic factor.

\begin{figure}
\begin{center}
\btabu{@{}cc}
\hspace{-2cm}
\includegraphics[height=7cm]{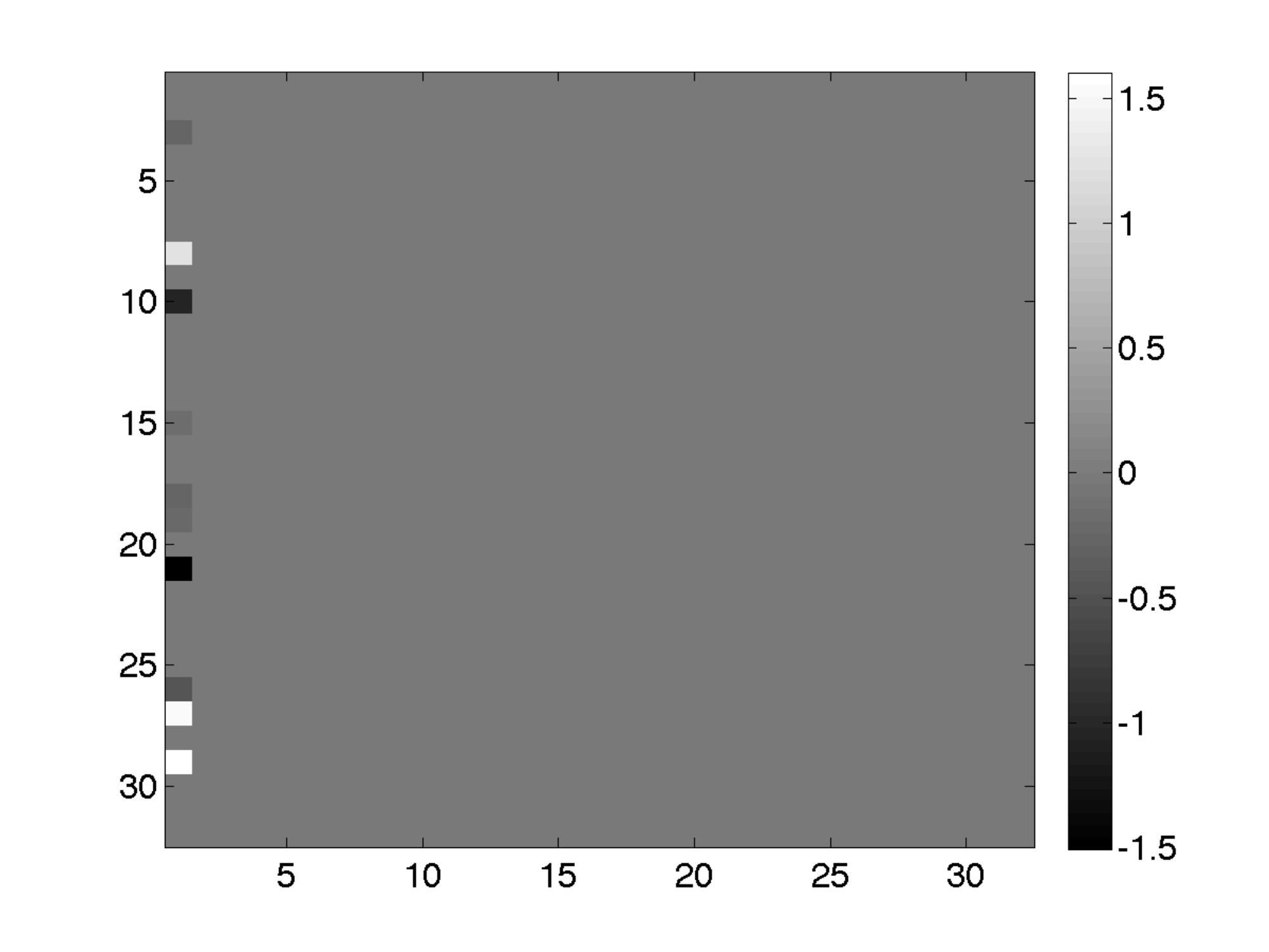} &
\includegraphics[height=7cm]{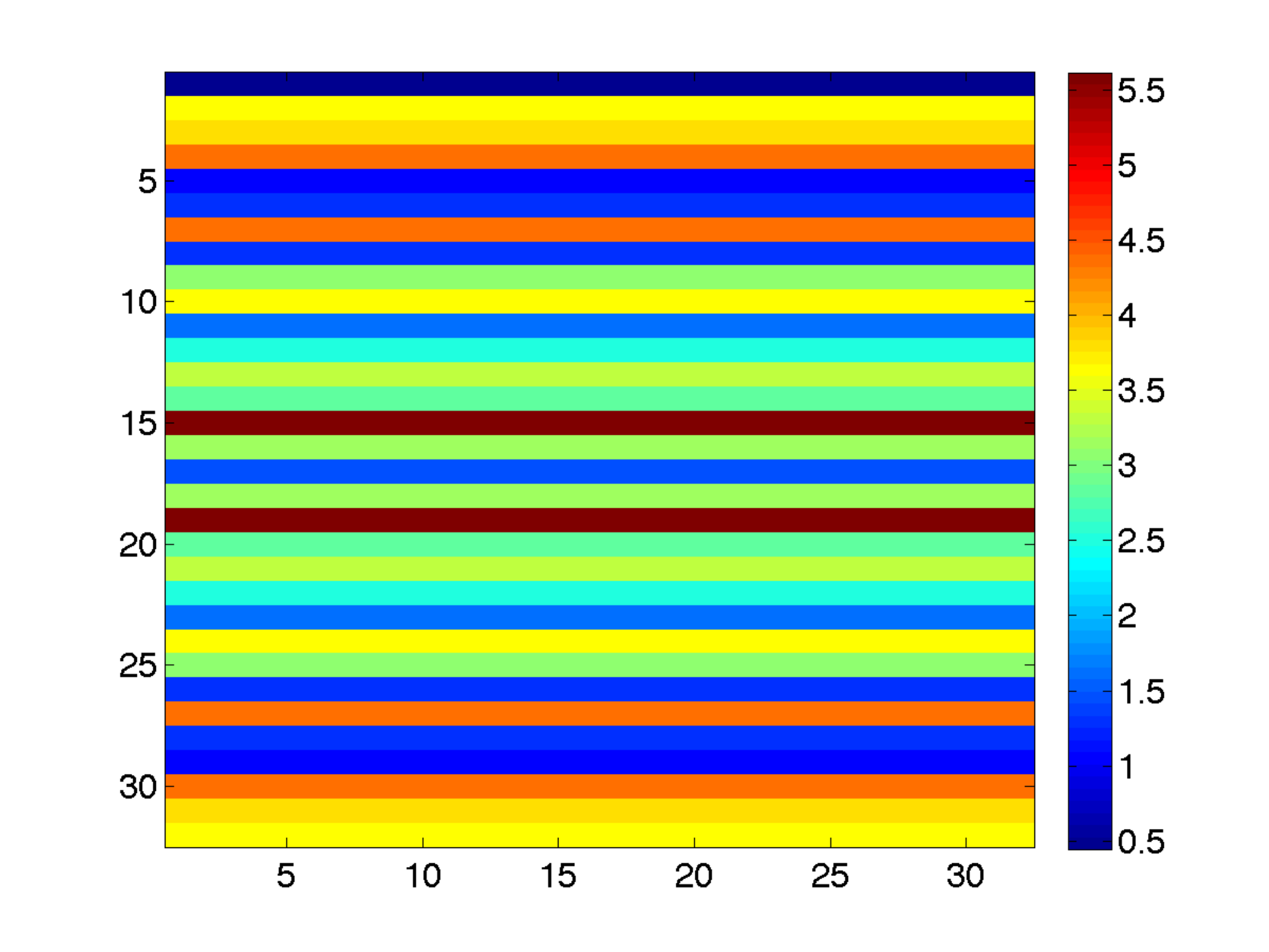} \\
\hspace{-2cm}
{\small (a)}&{\small (b)}
\etabu
\caption{\label{fig:pathoSignal} A pathological case where $n = 32 \times 32$ (a): The signal is $s$-sparse for $s=10$ and its support is concentrated on its first column. (b) Its 2D Fourier transform is constant along horizontal lines in the Fourier plane.   }
\end{center}
\end{figure}

%To sum up, in this part, we give a lower bound on the required number of blocks of measurements: we know that in such a setting we cannot expect to perform better with less than $O(s)$ blocks of measurements thanks to the following inequality
%$$ \inf_{\Delta , \Ab} \sup_{\zb \in \Sigma_S} d(\xb, \zb) \geq  \inf_{\Delta , \Ab} \sup_{\zb \in \Cbb_s} d(\xb, \zb) 
%$$
%where $\Delta$ represents any decoder, $\Ab$ the sampling matrix, and $d$ the distance chosen to be minimized between the exact solution and the reconstructed one.

\subsubsection{Relation to previous work\label{subsubsec:FourierWavelet}}

To the best of our knowledge, the only existing compressed sensing results based on blocks of measurements appeared in \cite{polak2012performance}. 
In this paragraph, we outline the differences between both approaches.

First, in our work, no assumption on the sign pattern of the non-zero signal entries is required. 
Furthermore, while the result in \cite{polak2012performance} only covers the case described in Paragraph \ref{subsubsec:partition} (i.e. partitions of orthogonal transforms), our work covers the case of overlapping blocks of measurements (see Paragraph \ref{subsubsec:overlap}), subsampled tight or continuous frames (see Paragraph \ref{subsubsec:frames}), and it can also be extended to the case of randomly generated blocks (see Paragraph \ref{subsubsec:gaussian}). 
Last but not least, the work \cite{polak2012performance} only deals with \emph{uniform} sampling densities which is well known to be of little interest when dealing with partially coherent matrices (see e.g. end of Paragraph \ref{subsec:isolmeas} for an edifying example).

Apart from those contextual differences, the comparison between the results in \cite{polak2012performance} and the ones in this paper is not straightforward. 
The criterion in \cite{polak2012performance} that controls the overall number of measurements $q$ depends on the following quantity:
$$ 
\Upsilon (\Ab_0 , S , \Bb) := \| \overline{\Bb_{S}} \|_{2 \rightarrow 1},
$$ 
where $\overline{\Bb_{S}}$ stands for the block restricted to the columns in $S$ with renormalized rows. 
The total number of measurements required in the approach \cite{polak2012performance} is 
\begin{align}\label{eq:PDG}
q_{PDG} &\geq C \Upsilon(\Ab_0 , S, \Bb) \max_{i,j} |\Ab_0(i,j) |^3 n^{3/2} \log(n)
\end{align}
which should be compared to our result
\begin{align}\label{eq:measnumber}
q &\geq  c p{\gamma(S)} \log \left( 4n \right) \log\left (12 \varepsilon^{-1}\right).
\end{align}
As shown in the previous paragraphs, the number \eqref{eq:measnumber} is sharp in various settings of interest, while \eqref{eq:PDG} is usually hard to explicitly compute or too large in the case of patially incoherent transforms. It therefore seems that our results should be preferred over those of \cite{polak2012performance}.

\section{Outlook}
We have introduced new sensing matrices that are constructed by stacking random blocks of measurements. 
Such matrices play an important role in applications since they can be implemented easily on many imaging devices. 
We have derived theorems that guarantee exact reconstruction using these matrices via $\ell_1$-minimization algorithms and outlined the crucial role of two properties: the \textit{extra} and \textit{intra} support block-coherences introduced in Definition \ref{def:quantities}. 
We have showed that our main result (Theorem \ref{thm:recovery}) is sharp in a few settings of practical interest, suggesting that it cannot be improved in the general case up to logarithmic factors.

Apart from those positive results, this work also reveals some limits of block sampling approaches. 
First, it seems hard to evaluate the \textit{extra} and \textit{intra} support block-coherences - except in a few particular cases - both analytically and numerically.
This evaluation is however central to derive optimal sampling approaches.
More importantly, we have showed in Paragraph \ref{subsubsec:limits} that not much could be expected from this approach in the specific setting where separable transforms and blocks consisting of lines of the acquisition space are used. 
Despite the peculiarity of such a dictionary, we believe that this result might be an indicator of a more general weakness of block sampling approaches. 
Since the best known compressed sensing strategies heavily rely on randomness (e.g. Gaussian measurements or uniform drawings of Fourier atoms), one may wonder whether the more rigid sampling patterns generated by block sampling approaches have a chance to provide decent results.
It is therefore legitimate to ask the following question: is it reasonable to use variable density sampling with pre-defined blocks of measurements in compressed sensing? 

Numerical experiments indicate that the answer to this question is positive. 
For instance, it is readily seen in Figure \ref{fig:reconstruction} (a,b,c) and (j,k,l), that block sampling strategies can produce comparable results to acquisitions based on isolated measurements. 
The first potential explanation to this phenomenon is that $\gamma(S)$ is low for the dictionaries chosen in those experiments.
However, even acquisitions based on horizontal lines in the Fourier domain (see Figure \ref{fig:reconstruction} (d,e,f)) produce rather good reconstruction results while Proposition \ref{prop:2DFour2} seems to indicate that this strategy is doomed.

This last observation suggests that a key feature is missing in our study to fully understand the potential of block sampling in applications.
Recent papers \cite{adcock2013breaking,adcock2014quest} highlight the central role of \textit{structured sparsity} to explain the practical success of compressed sensing. A very promising perspective is therefore to couple the ideas of structured sparsity in \cite{adcock2013breaking,adcock2014quest} and the ideas of block sampling proposed in this paper to finely understand the results in Figure \ref{fig:reconstruction} and perhaps design new optimal and applicable sampling strategies.

%A potential explanation to . Moreover, the significant differences in terms of PSNR between the different sampling patterns, suggesting that $\gamma(S)$ . 
%By comparing the reconstruction results between the different sampling strategies, one can nevertheless see significant differences in the reconstruction quality. This could be explained 

\begin{figure}
\begin{center}
\btabu{@{}ccc}
\includegraphics[height=4.5cm]{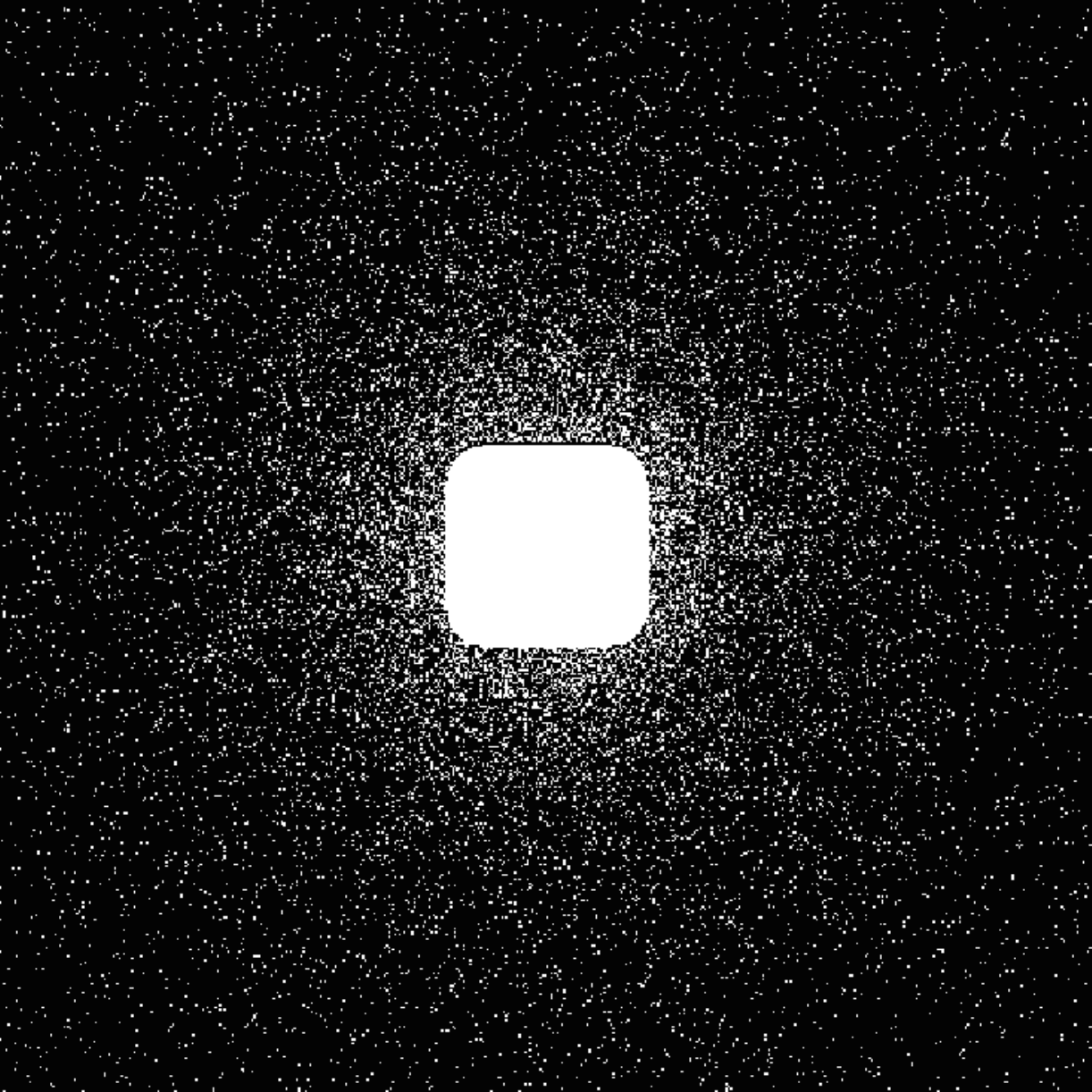} &
\includegraphics[height=4.5cm]{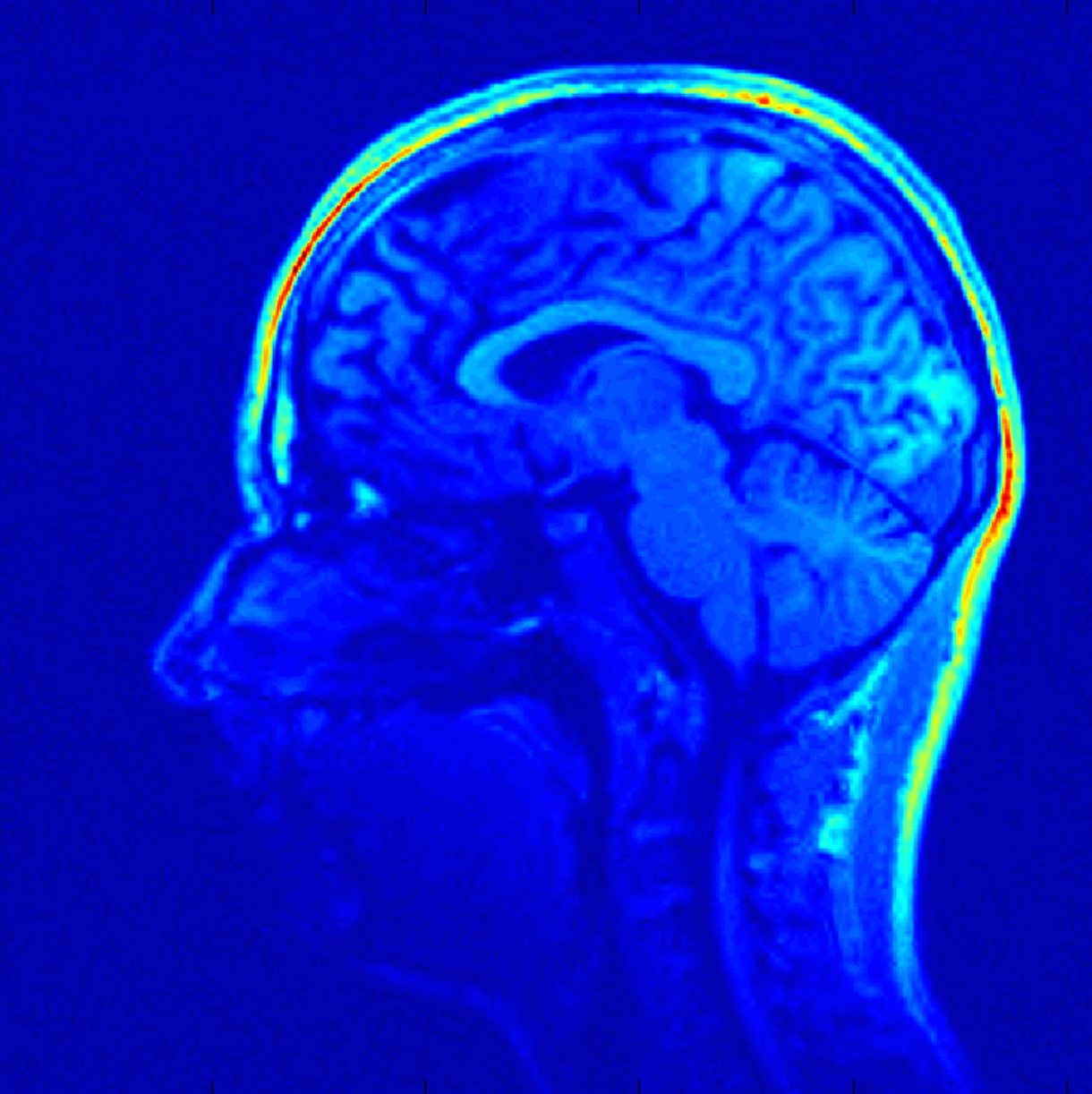} &
\includegraphics[height=4.5cm]{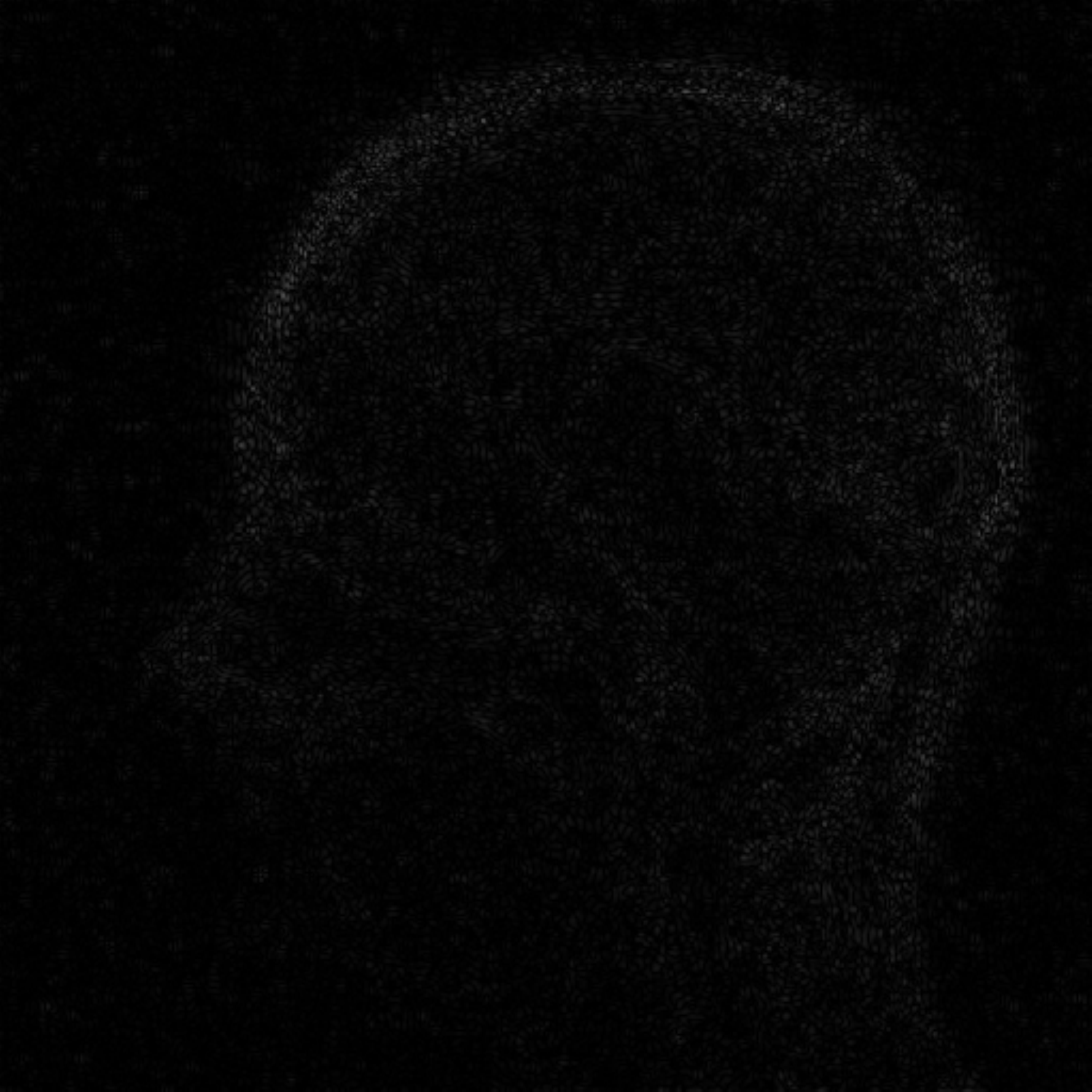} \\
{\small (a)}&{\small (b) PSNR = 40 dB} & {\small (c)} \\
\includegraphics[height=4.5cm]{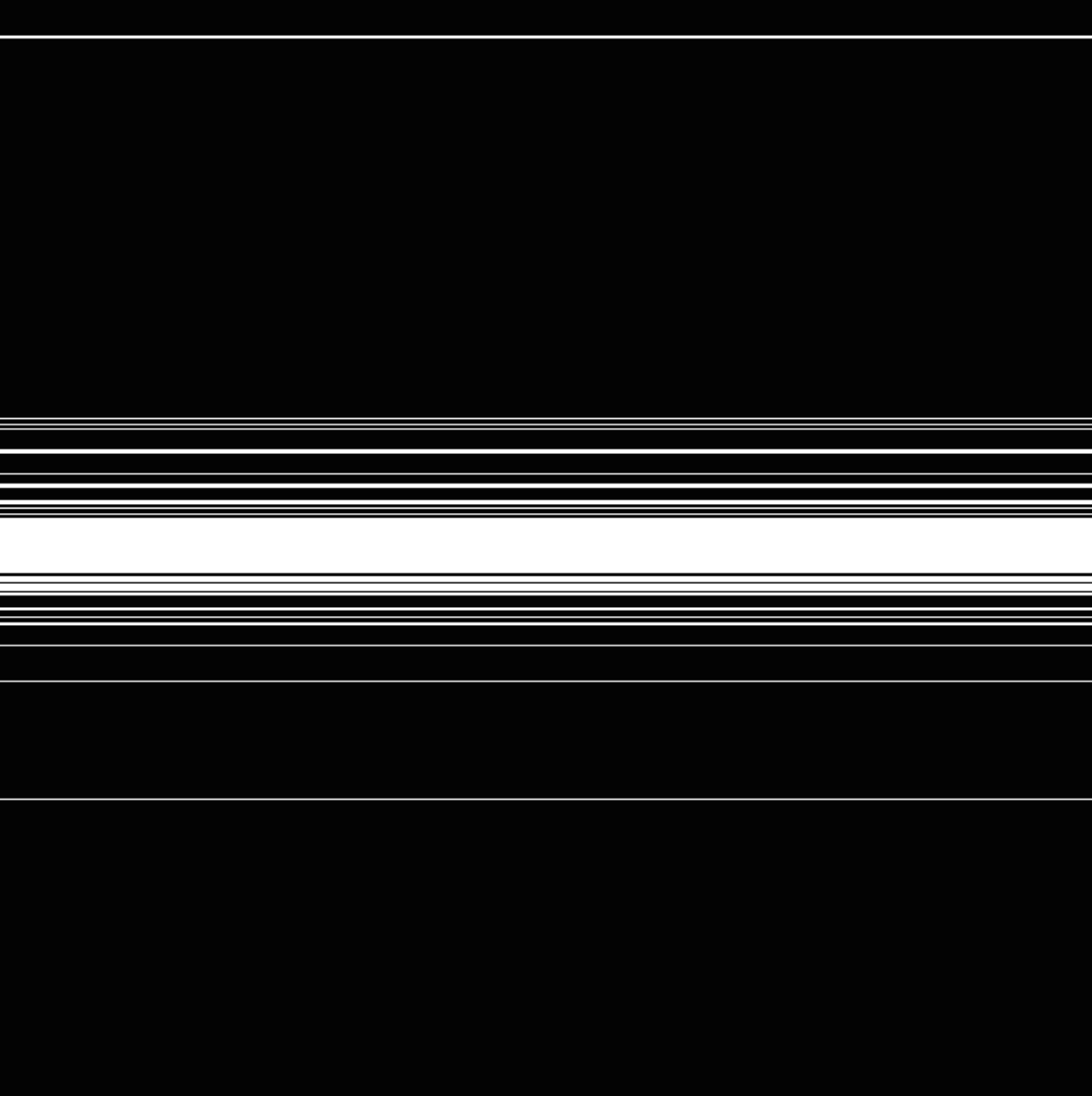} &
\includegraphics[height=4.5cm]{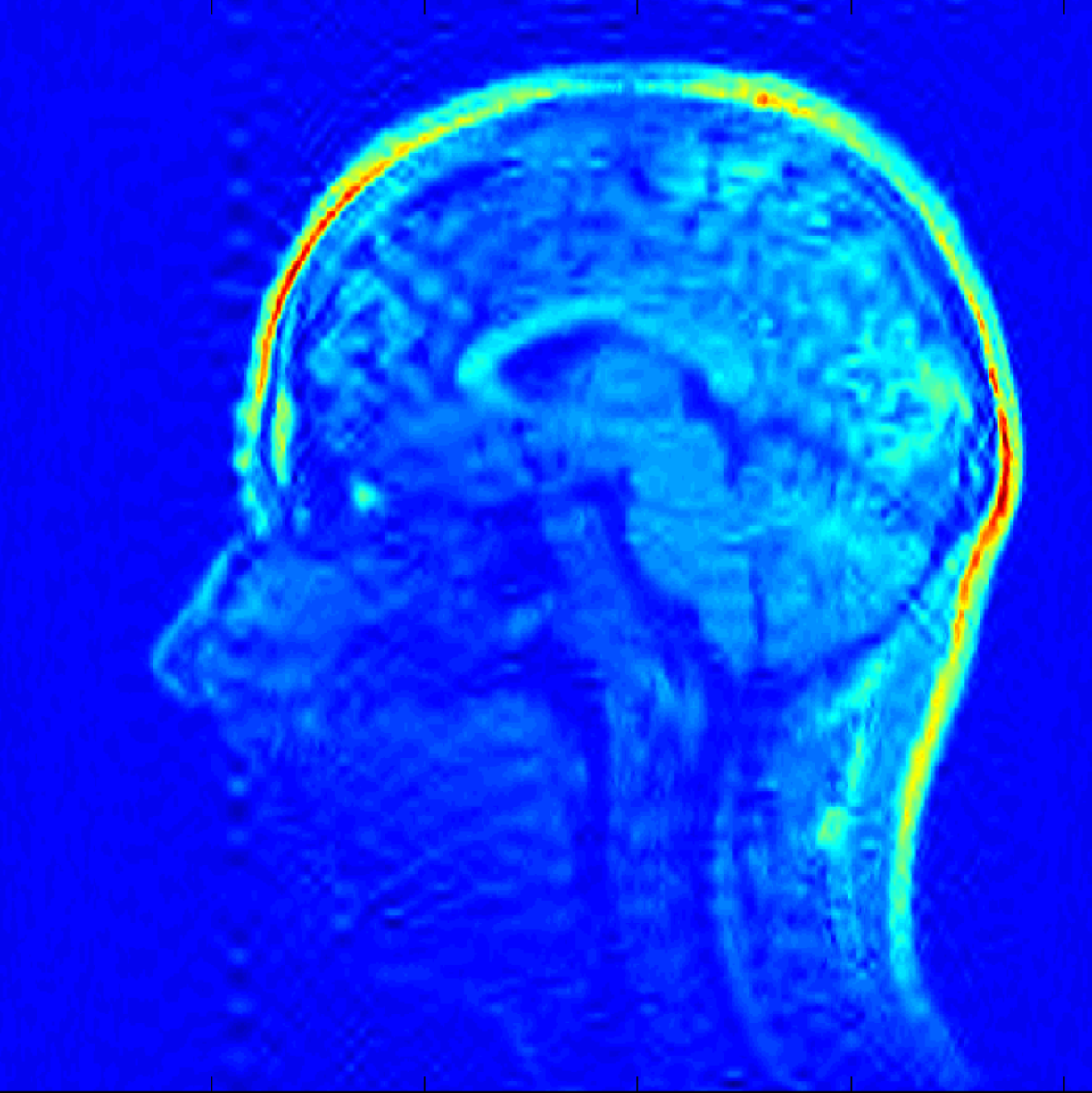} &
\includegraphics[height=4.5cm]{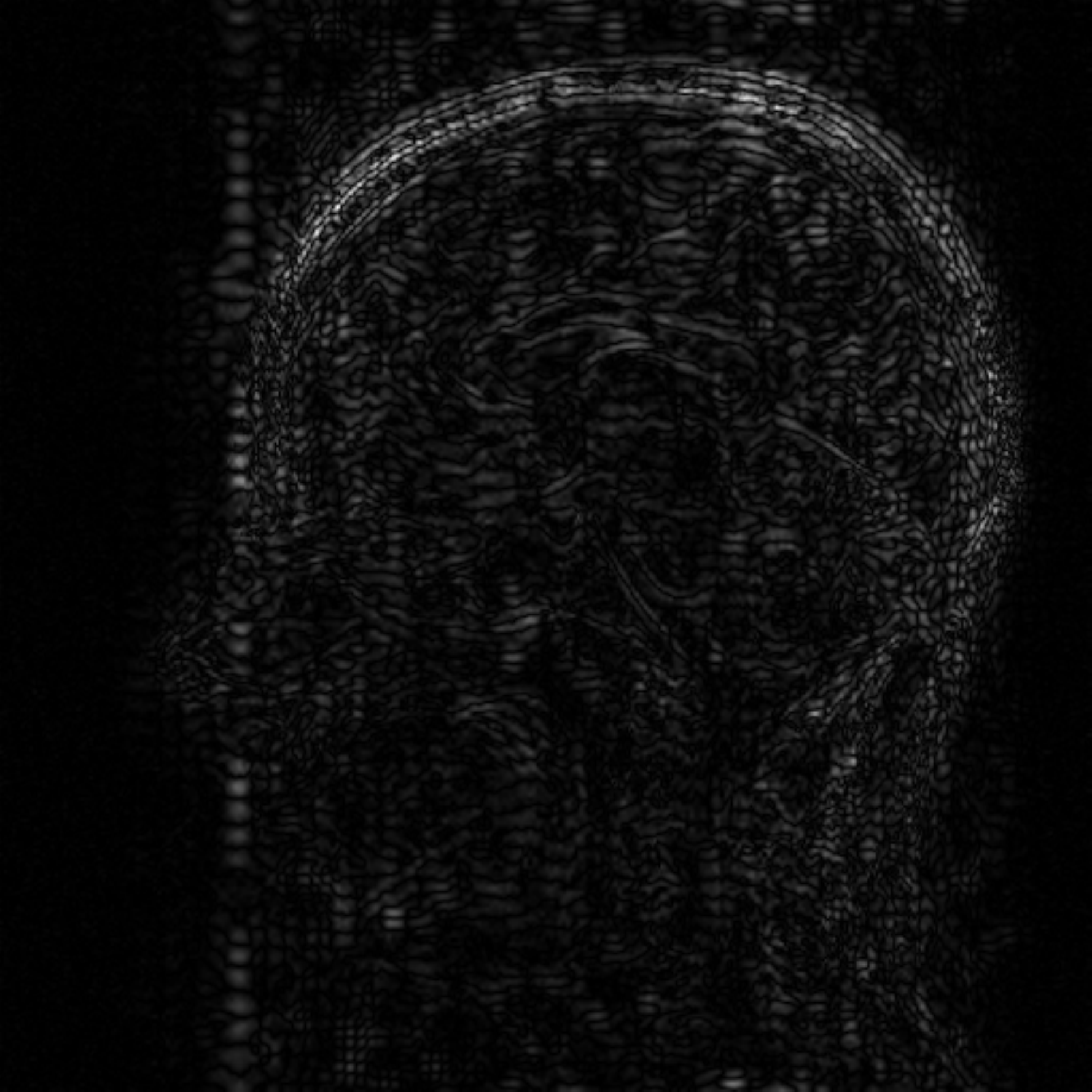} \\
{\small (d)}&{\small (e) PSNR = 32.79 dB} & {\small (f)}\\
\includegraphics[height=4.5cm]{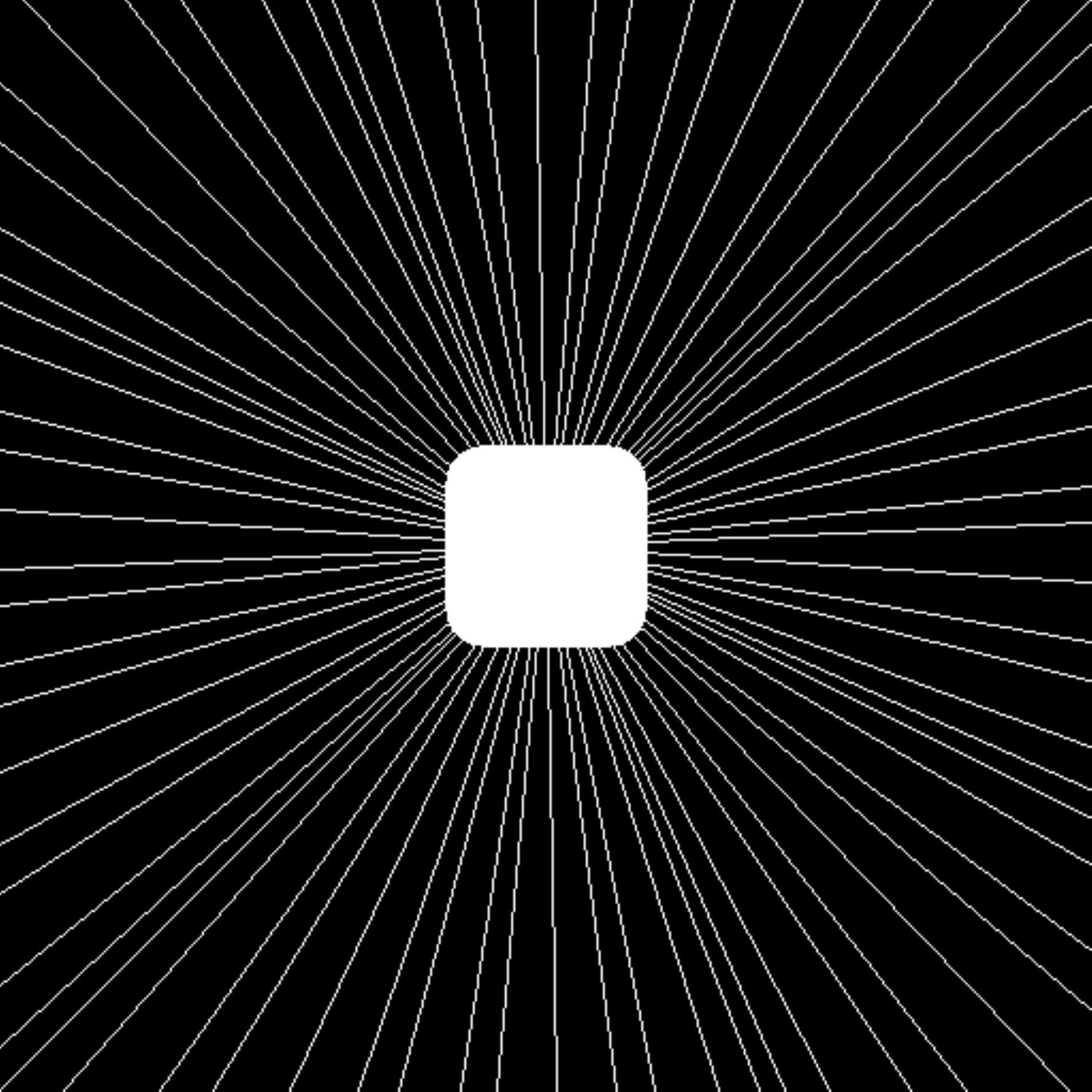}&
\includegraphics[height=4.5cm]{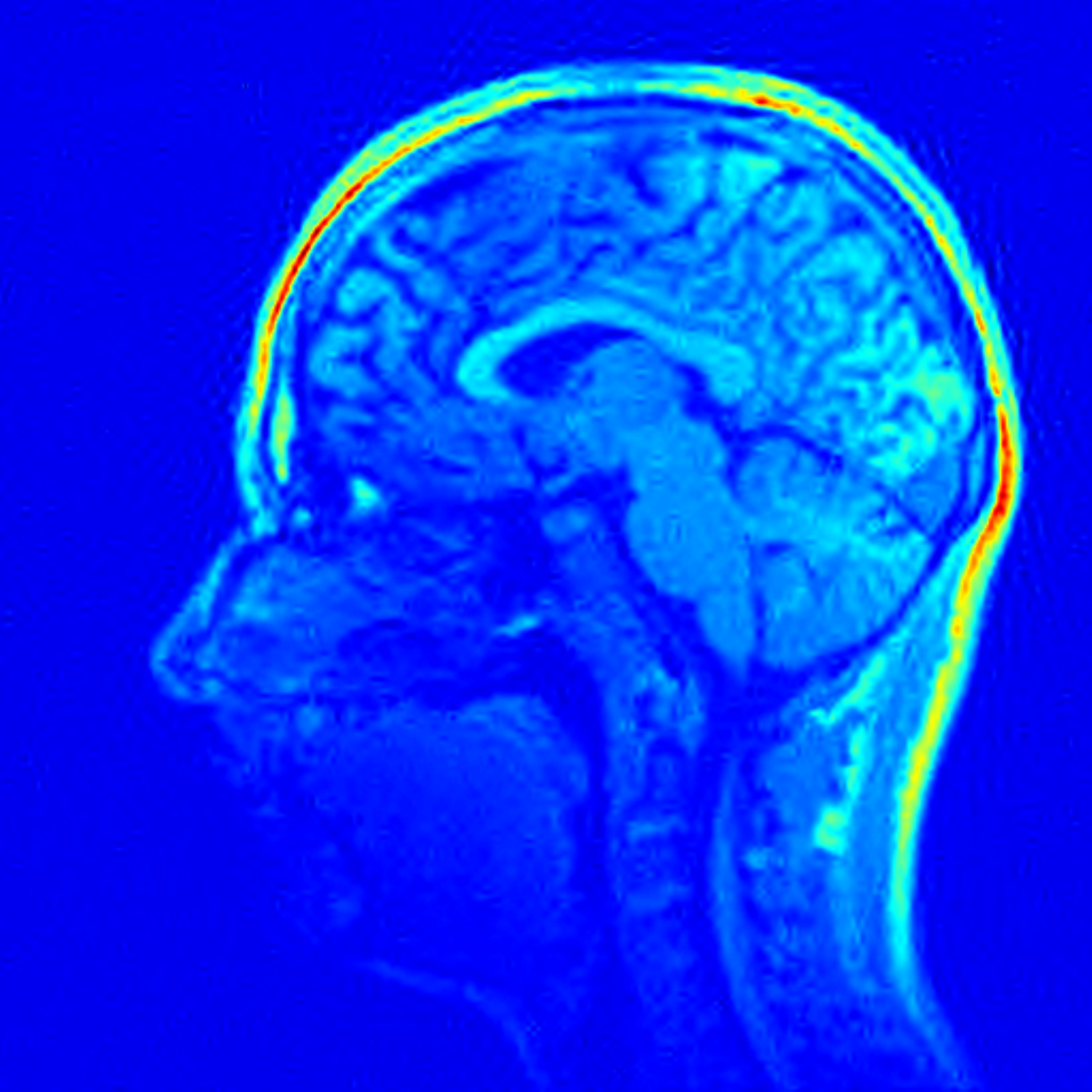}&
\includegraphics[height=4.5cm]{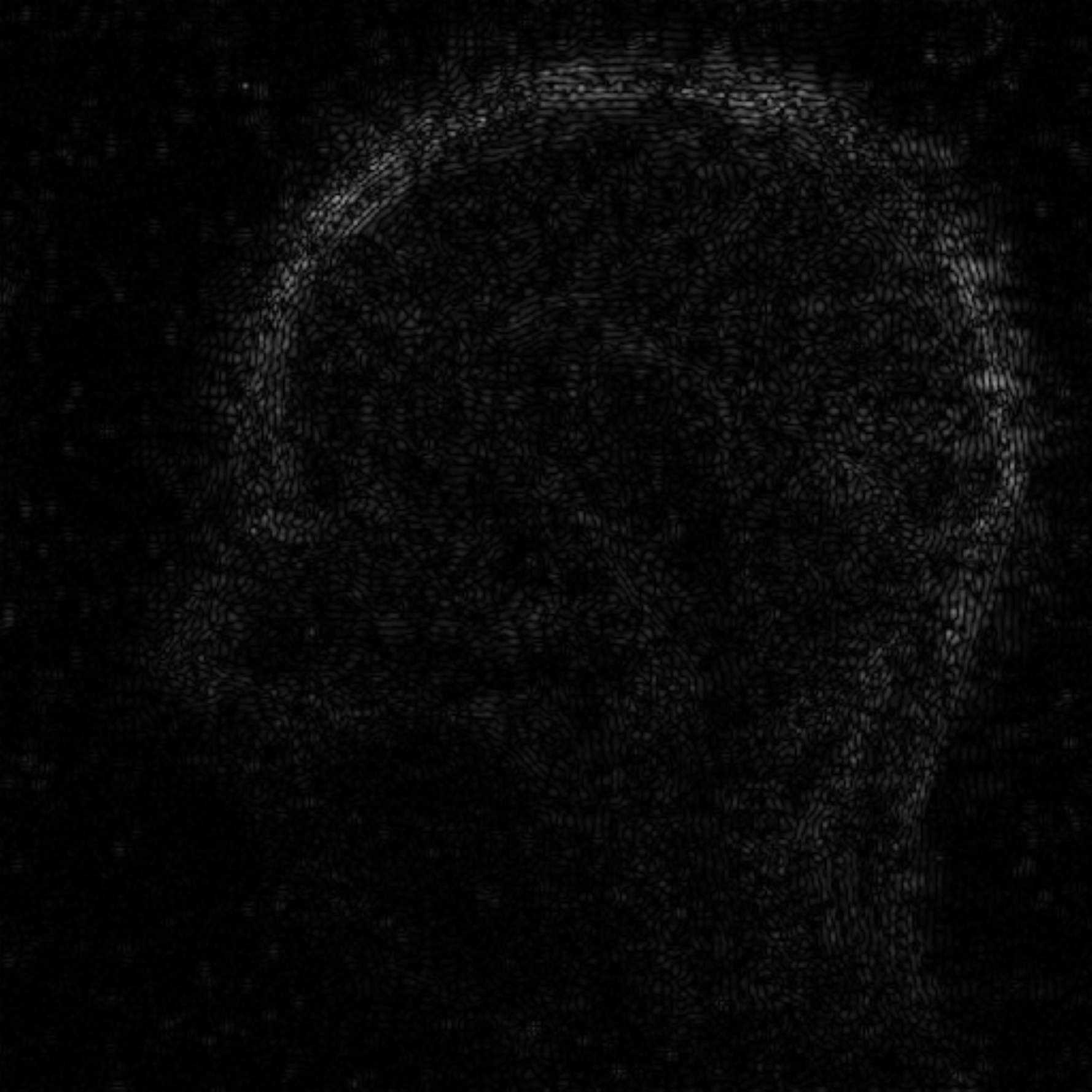} \\
{\small (g)}&{\small (h) PSNR = 36.34 dB } & {\small (i)} \\
\includegraphics[height=4.5cm]{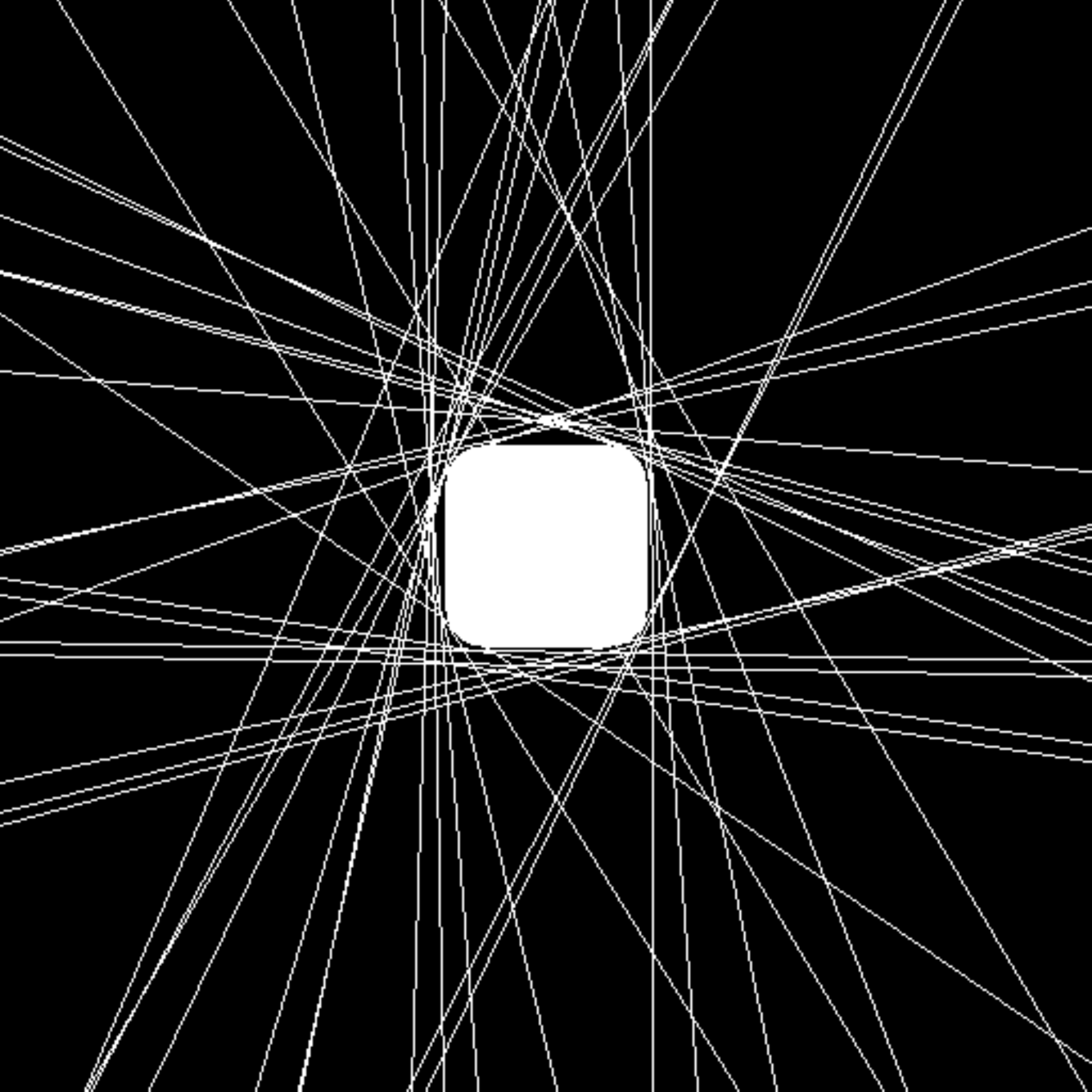} &
\includegraphics[height=4.5cm]{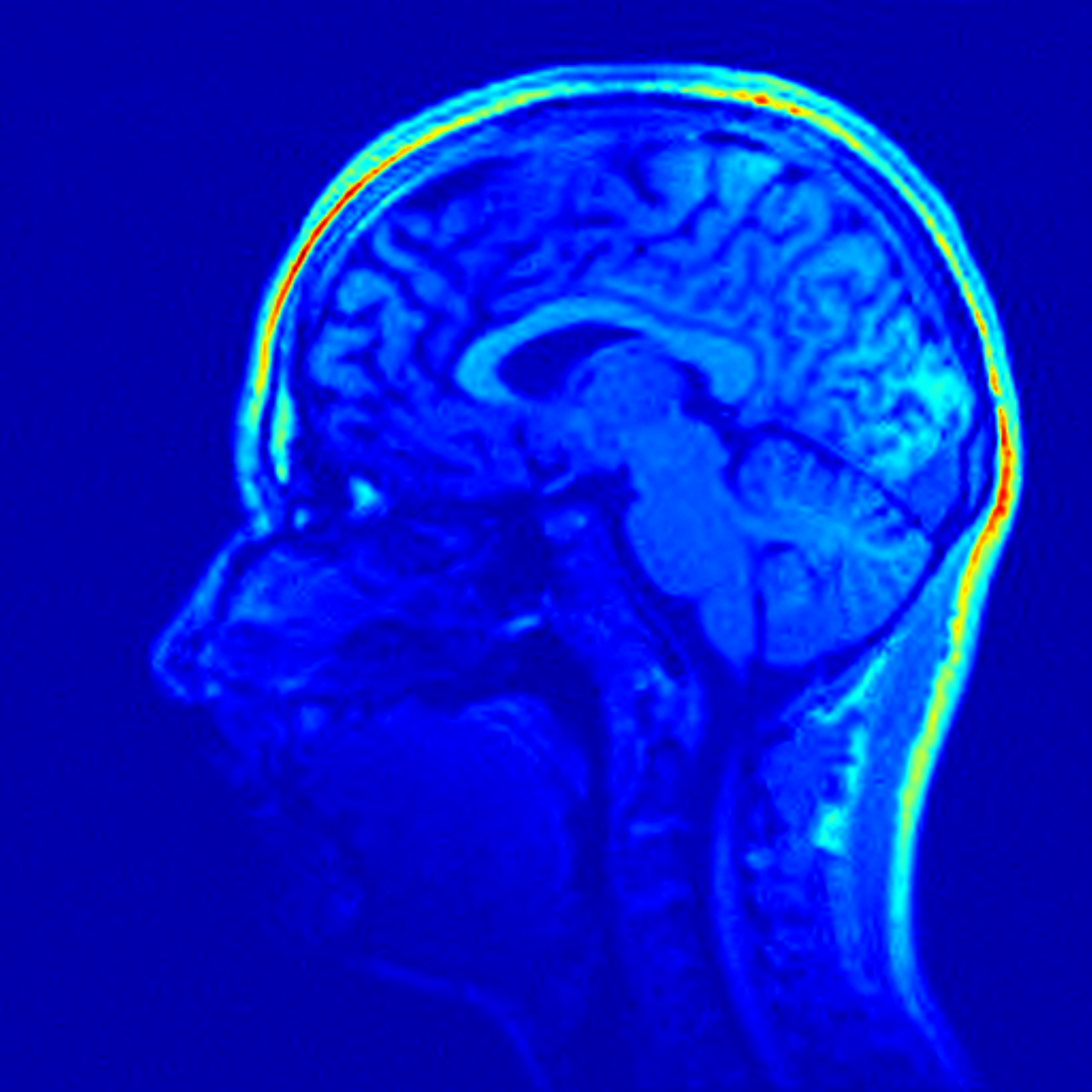} &
\includegraphics[height=4.5cm]{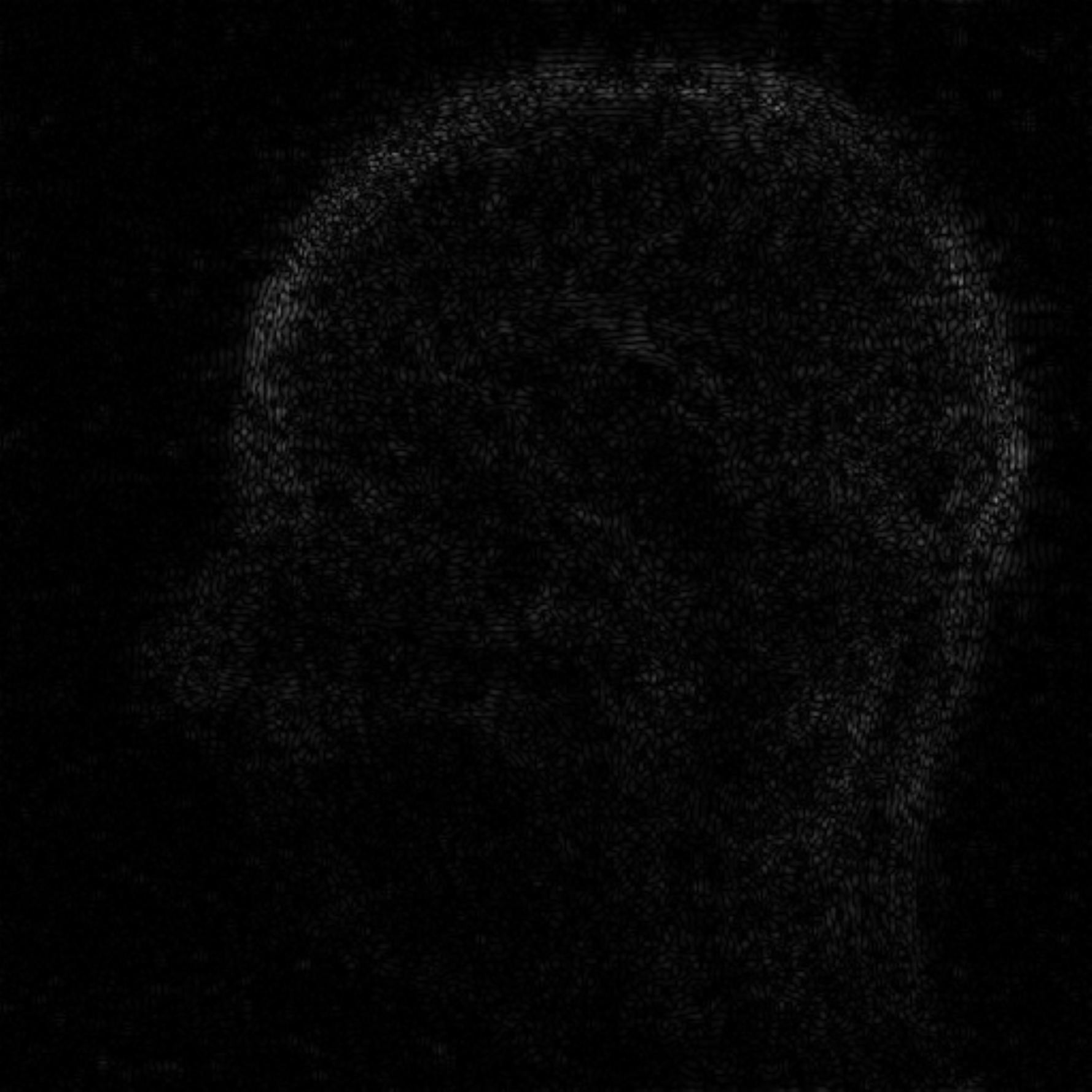} \\
{\small (j)}&{\small (k) PSNR = 38.99 dB} & {\small (l)}
\etabu
\caption{\label{fig:reconstruction} Reconstruction results using different sampling strategies. Each sampling pattern contains $10\%$ of the total number of possible measurements. From top to bottom: measurements drawn independently at random with a radial distribution - horizontal lines in the Fourier domain - deterministic radial sampling - heuristic method proposed in \cite{boyer2014algorithm}. From left to right: sampling scheme - corresponding reconstruction - difference with the reference (the same colormap is used in every experiment).}
\end{center}
\end{figure}

\clearpage
\appendix

\section{Bernstein's inequalities}

\begin{thmchapter}[Scalar Bernstein Inequality]
\label{theo:scalBern}
Let $x_1, \hdots , x_m$ be independent random variables such that $|x_\ell |\leq K$ almost surely for every $\ell \in \{ 1, \hdots , m \}$. Assume that $\Ebb | x_\ell |^2 \leq \sigma^2_\ell$ for $\ell \in \{ 1, \hdots , m \}$. Then for all $t>0$,
\begin{align*}
\Pbb \left( \left| \sum_{\ell=1}^m x_\ell \right| \geq t  \right) \leq 2 \exp\left(  -\frac{t^2/2}{\sigma^2 + Kt/3}\right),
\end{align*}
with $\sigma^2 \geq \sum_{\ell=1}^m \sigma_\ell^2$.
\end{thmchapter}

\begin{thmchapter}[Rectangular Matrix Bernstein Inequality]
 \label{theo:MatBern}
\cite[Theorem 1.6]{tropp2012user}

 Let $(\Zb_k)_{1 \leq k \leq m}$ be a finite sequence of rectangular independent random matrices of dimension $d_1\times d_2$. Suppose that $\Zb_k$ is such that $\Ebb \Zb_k = 0$ and $\| \Zb_k \|_{2\rightarrow 2} \leq K$ a.s.\ for some constant $K > 0$ that is independent of $k$. Define 
$$ \sigma^2 \geq \max\left( \left\|\sum_{k = 1}^{m} \Ebb \Zb_k \Zb_k^* \right\|_{2\rightarrow 2} , \left\|\sum_{k = 1}^{m} \Ebb \Zb_k^* \Zb_k \right\|_{2\rightarrow 2}\right).
$$
Then, for any $t > 0$, we have that
$$ \Pbb \left( \left\|\sum_{k = 1}^{m} \Zb_k  \right\|_{2\rightarrow 2} \geq t \right) \leq \left( d_1 + d_2 \right) \exp\left( -\frac{t^2/2}{\sigma^2 + Kt/3} \right)
$$

\end{thmchapter}

\begin{thmchapter}[Vector Bernstein Inequality (V1)] 
\label{theo:VectBern2}
\cite[Theorem 2.6]{candes2011probabilistic}
Let $\left( \yb_k \right)_{1\leq k \leq m}$ be a finite sequence of  independent and identically distributed random vectors of dimension $n$. Suppose that  $\Ebb \yb_{1} = 0$ and $\| \yb_{1} \|_{2} \leq K$ a.s.\ for some constant $K > 0$ and set $\sigma^2 \geq \sum_k \Ebb \| \yb_k \|_2^2$.  Let $Z = \left\| \sum_{k = 1}^{m} \yb_k \right\|_2 $. Then, for any $0< t \leq \sigma^2/K $, we have that
$$
 \Pbb \left( Z \geq  t \right) \leq  \exp\left( -\frac{\left( t / \sigma - 1 \right)^2}{ 4} \right) \leq \exp \left( - \frac{ t^2}{8\sigma^2 } +  \frac{1}{4}\right),
$$
where $ \Ebb Z^2  =  \sum_{k = 1}^{m}  \Ebb \| \yb_{k} \|_{2}^{2} = m   \Ebb \| \yb_{1} \|_{2}^{2}$.
\end{thmchapter}

\begin{thmchapter}[Vector Bernstein Inequality (V2)] 
\label{theo:VectBern}
\cite[Corollary 8.44]{foucart2013mathematical}
Let $\left( \yb_k \right)_{1\leq k \leq m}$ be a finite sequence of  independent and indentically distributed random vectors of dimension $n$. Suppose that  $\Ebb \yb_{1} = 0$ and $\| \yb_{1} \|_{2} \leq K$ a.s.\ for some constant $K > 0$. Let $Z = \left\| \sum_{k = 1}^{m} \yb_k \right\|_2 $. Then, for any $t > 0$, we have that
$$
 \Pbb \left( Z \geq \sqrt{\Ebb Z^2} + t \right) \leq  \exp\left( -\frac{t^2/2}{ \Ebb Z^2 +  2K\sqrt{\Ebb Z^2}  +  Kt/3} \right),
$$
where $ \Ebb Z^2  =  \sum_{k = 1}^{m}  \Ebb \| \yb_{k} \|_{2}^{2} = m   \Ebb \| \yb_{1} \|_{2}^{2}$.
Note that the previous inequality still holds by replacing $\Ebb Z^2$ by $\sigma^2$ where $\sigma^2 \geq \Ebb Z^2$.
%
%$$
%Z = \left\| \sum_{\ell = 1}^{m} \Yb_\ell \right\|_2 \mbox{ and } \sigma^2 = \sup_{\|\xb \|_2 \leq 1} \Ebb \left| \left\langle  \xb , \Yb_{\ell} \right\rangle \right|.
%$$
%Then, for any $t > 0$, we have that
%$$ \Pbb \left( Z \geq \sqrt{\Ebb Z^2} + t \right) \leq  \exp\left( -\frac{t^2/2}{m\sigma^2 +  2K\sqrt{\Ebb Z^2}  +  Kt/3} \right).
%$$
%Using that $\sigma^2 \leq  \Ebb \left\| \sum_\ell \Yb_\ell \right\|_2^2 =  \Ebb Z^2$, one also has  that for any $t > 0$,
%$$
% \Pbb \left( Z \geq \sqrt{\Ebb Z^2} + t \right) \leq  \exp\left( -\frac{t^2/2}{m \Ebb Z^2 +  2K\sqrt{\Ebb Z^2}  +  Kt/3} \right).
%$$
%
\end{thmchapter}

\section{Estimates: auxiliary results} \label{sec:aux}

Let $S$ be the support of the signal to be reconstructed such that $|S|=s$. Note that the isotropy condition \eqref{condIsotropy}  ensures that the following properties hold 
\begin{enumerate} 
\item $\Ebb \left( {\Bb^* \Bb} \right) = \Id_n$ and $\Ebb \left( \Bb_{S}^* \Bb_{S} \right) = \Id_s$. 
\item for any vector $\wb \in \Cbb^s$,  $\Ebb \left[  \Bb_{S} \wb   \right]^2 = \| \wb \|_2^2$.
\item for any $i \in {S^c}$, $ \Ebb \left( \Bb_{S}^* \Bb \eb_i\right)= 0$.
\end{enumerate}
The above properties will be repeatedly used in the proof of the following lemmas.

\begin{lemme}
\label{lem:localIsometry}
Let $S \subset \{1,\hdots , n\}$ be of cardinality of $s$. Then, for any $\delta >0$, one has that
\begin{align}
\tag{E1}
\label{lem:E1eq}
\displaystyle \Pbb \left( \| \Ab_S^* \Ab_S - \Id_{s} \|_{2\rightarrow 2} \geq \delta \right) &\leq 2 s \exp \left( -  \frac{m \delta^2/2}{\mu_1(S)  +   \max(\mu_1(S)-1 , 1)\delta /3)} \right).
\end{align}
\end{lemme}
\begin{proof}
We decompose the matrix $\Ab_S^* \Ab_S-\Id_{s}$ as 
$$
\Ab_S^* \Ab_S- \Id_{s}= \frac{1}{m} \sum_{k=1}^m \left( \Bb_{k,S}^{*} \Bb_{k,S} - \Id_{s} \right) =\frac{1}{m} \sum_{k=1}^m X_{k},
$$
where $X_{k} :=\displaystyle \left( \Bb_{k,S}^{*} \Bb_{k,S} - \Id_{s}\right)$.
It is clear that $\Ebb X_{k} = 0$, and since 
$
\left\| \Bb_{k,S}^{*} \Bb_{k,S}\right\|_{2\rightarrow 2} \leq  \mu_1(S),
$
we have that
$$\|X_{k} \|_{2\rightarrow 2} =  \max \left( \left\| \Bb_{k,S}^{*} \Bb_{k,S}\right\|_{2\rightarrow 2} -1 ,1\right)\leq  \max( \mu_1(S)-1 , 1).$$
Lastly, we remark that
\begin{align*}
0 \quad \preceq \quad  \Ebb X_{k}^2 = \Ebb\left[  \Bb_{k,S}^{*} \Bb_{k,S} \right]^2 - \Id_{s} \quad &\preceq \quad \Ebb \left\| \Bb_{k,S}^{*} \Bb_{k,S} \right\|_2 \Bb_{k,S}^{*} \Bb_{k,S} \preceq   \mu_1(S)  \Id_{s}.
\end{align*}
Therefore,
$\sum_{k=1}^{m} \Ebb X_k^2 \preceq  m  \mu_1(S) \Id_{s} $ which implies that 
$
\left\| \sum_{k=1}^{m} \Ebb X_k^2 \right\|_2 \leq m  \mu_1(S). 
$
Hence, inequality \eqref{lem:E1eq} follows immediately from  Bernstein's inequality for random matrices (see Therorem \ref{theo:MatBern}).
\end{proof}

\begin{lemme}
\label{lem:lowDistortion}
Let $S \subset \{1, \hdots, n\}$, such that $|S|=s$. Let $\wb$ be a vector in $\Cbb^s$. Then, for any $t > 0$, one has that
\begin{align}
\tag{E2}
\label{lem:E2eqt}
\Pbb & \left( \left\| \left( \Ab_S^{*} \Ab_S -\Id_s  \right) \wb \right\|_2\geq \left( \sqrt{\frac{\mu_1(S) -1}{m} } + t\right) \|\wb\|_2 \right) \\
& \leq \exp \left( - \frac{m t^2/2}{(\mu_1(S) -1) + 2 \sqrt{\frac{\mu_1(S)-1}{m}} \mu_1(S) + \mu_1(S) t/3} \right).\nonumber
\end{align}
\end{lemme}

\begin{proof}
Without loss of generality we may assume that $\| \wb \|_2 =1$. We remark that
$$
\left( \Ab_S^* \Ab_S - \Id_s \right) \wb_S =  \frac{1}{m} \sum_{k=1}^m \left( \Bb_{k,S}^{*} \Bb_{k,S} - \Id_s\right) \wb = \frac{1}{m} \sum_{k=1}^m  \yb_k,  
$$
where $ \yb_k = \left( \Bb_{k,S}^{*} \Bb_{k,S}- \Id_s \right) \wb$ is a random vector with zero mean.  Simple calculations yield that
\begin{align*}
 \left\| \frac{1}{m} \yb_k \right\|_2^2 &= \frac{1}{m^2} \left(  \wb^* \left( \Bb_{k,S}^{*} \Bb_{k,S} \right)^2 \wb   -2 \wb^* \Bb_{k,S}^* \Bb_{k,S} \wb +  \wb^*\wb \right) \\
 &\leq \frac{1}{m^2} \left( \mu_1(S) \wb^*  \Bb_{k,S}^{*} \Bb_{k,S}  \wb   -2 \wb^* \Bb_{k,S}^* \Bb_{k,S} \wb + 1\right) \\
 &= \frac{1}{m^2} \left( \left( \mu_1(S) -2 \right)\wb^* \Bb_{k,S}^* \Bb_{k,S} \wb  + 1  \right) \\
 &\leq \frac{1}{m^2} \left( \left( \mu_1(S) -2 \right) \mu_1(S) \| \wb\|_2^2  + 1  \right) = \frac{1}{m^2} \left( \left( \mu_1(S) -2 \right) \mu_1(S) + 1  \right)  \\
 &\leq \frac{1}{m^2}  \left( \mu_1(S) - 1  \right)^2 \leq \frac{1}{m^2}  \mu_1^2(S).
\end{align*}
Now, let us define 
$
Z = \left\| \frac{1}{m} \sum_{k=1}^m  \yb_k\right\|_2.
$
By independence of the random vectors $\yb_{k}$, it follows that
\begin{align*}
\Ebb \left[ Z^2 \right] &= \frac{1}{m} \Ebb \left\|   \yb_1\right\|^2_2 = \frac{1}{m} \Ebb \left[ \left\langle  \Bb_S^{*} \Bb_S  \wb ,   \Bb_S^{*} \Bb_S\wb \right\rangle - 2 \left\langle   \Bb_S^{*} \Bb_S \wb, \wb \right\rangle + \left\langle \wb , \wb \right\rangle   \right] \\
&= \frac{1}{m} \Ebb \left[ \left\langle \left(  \Bb_S^{*} \Bb_S \right)^2 \wb , \wb \right\rangle - 2 \left\|    \Bb_S  \wb \right\|_2^2 + 1   \right]. 
\end{align*}
To bound the first term in the above equality, one can write
\begin{align*} 
\Ebb &\left[ \left\langle \left(  \Bb_S^{*} \Bb_S \right)^2 \wb , \wb \right\rangle \right] = \left\langle \Ebb\left[\left(  \Bb_S^{*} \Bb_S \right)^2\right] \wb , \wb \right\rangle \\
&\leq \mu_1(S) \left\langle \Ebb\left[\left(  \Bb_S^{*} \Bb_S \right)\right] \wb , \wb \right\rangle \leq   \mu_1(S) \|\wb\|_2^2 = \mu_1(S). 
\end{align*}
One immediately has that
$
\Ebb \left\langle \Bb_S  \wb ,  \Bb_S   \wb \right\rangle  =  \| \wb \|_2^2 = 1.
$
Therefore, one finally obtains that
$$  \Ebb \left[ Z^2 \right] \leq \frac{\mu_1(S) -1}{m}.
$$
Using the above upper bounds, namely $ \left\| \frac{1}{m} \yb_k  \right\|_2 \leq  \frac{\mu_1(S) }{m} $ and $\Ebb \left[ Z^2 \right] \leq \frac{\mu_1(S) -1}{m}$, the result of the lemma is thus a consequence of the Bernstein's inequality for random vectors  (see Theorem \ref{theo:VectBern}), which completes the proof. \end{proof}

\begin{lemme}
\label{lem:distortionInf2}
Let $S \subset \{1, \hdots, n\}$, such that $|S|=s$. Let $\vb$ be a vector of $\Cbb^s$. Then we have
\begin{align}
\tag{E3}
\label{lem:E3eq}
\Pbb\left( \left\| \Ab_{S^c}^* \Ab_S \vb \right\|_\infty \geq t \| \vb\|_2 \right) \leq 4n \exp\left( -\frac{mt^2/4}{\frac{\mu_3(S)}{s}  + \frac{\mu_2(S)}{\sqrt{s}} t/3} \right).
\end{align}
\end{lemme}

\begin{proof}
Suppose without loss of generality that $\| \vb\|_2=1$.
Then,
\begin{align*} 
 \left\| \Ab_{S^c}^* \Ab_S \vb \right\|_\infty &= \max_{i \in S^c} \left\langle \eb_i ,  \Ab^* \Ab_S \vb \right\rangle = \max_{i \in S^c } \frac{1}{m} \sum_{k = 1}^m \left\langle \eb_i ,  \Bb_{k}^* \Bb_{k,S} \vb \right\rangle.
\end{align*}
Let us define $Z_k = \frac{1}{m} \left\langle \eb_i ,  \Bb_{k}^* \Bb_{k,S} \vb \right\rangle$. Note that $\Ebb Z_k =0$.
From the Cauchy-Schwarz inequality, we get
\begin{align*}
|Z_k | &= \left| \frac{1}{m} \left\langle \eb_i ,  \Bb_{k}^* \Bb_{k,S} \vb \right\rangle \right|  = \left| \frac{1}{m}  \vb^* \Bb_{k,S}^*   (\Bb_{k} \eb_i)  \right|  \leq \frac{1}{m} \| \vb\|_2  \|\Bb_{k,S}^*   (\Bb_{k} \eb_i) \|_2  \leq    \frac{1}{m} \frac{\mu_2(S)}{\sqrt{s}}.
\end{align*}
Furthermore,
\begin{align*}
\Ebb | Z_k |^2 &= \frac{1}{m^2} \Ebb \left\langle (\Bb_{k} \eb_i ) , \Bb_{k,S}\vb  \right\rangle^2 \\
&\leq\frac{1}{m^2}  \vb^* \Ebb \left[ \Bb_S^{*} \left( \Bb \eb_i \right) \left( \Bb \eb_i \right)^* \Bb_S   \right] \vb
\\
&\leq \frac{1}{m^2}  \max_{i \in S^c} \left\| \Ebb \left[ \Bb_S^{*} \left( \Bb \eb_i \right) \left( \Bb \eb_i \right)^* \Bb_S   \right] \right\|_{2\rightarrow 2}
= \frac{1}{m^2} \frac{\mu_3(S)}{s}.
\end{align*}
Using Bernstein's inequality \ref{theo:scalBern} for complex random variables, we end to
\begin{align*}
\Pbb &\left( \frac{1}{m} \left| \sum_{k = 1}^m  \left\langle \eb_i ,  \Bb_{k}^* \Bb_{k} \vb \right\rangle \right| \geq t \right) \\
&\leq  \Pbb \left( \frac{1}{m} \left|  \sum_{k = 1}^m \text{Re} \left\langle \eb_i ,  \Bb_{k}^* \Bb_{k} \vb \right\rangle \right| \geq t / \sqrt{2}\right)  + \Pbb \left( \frac{1}{m} \left|  \sum_{k = 1}^m \text{Im} \left\langle \eb_i ,  \Bb_{k}^* \Bb_{k} \vb \right\rangle \right| \geq t /\sqrt{2} \right)
\\
&\leq 4 \exp\left( -\frac{mt^2/4}{\frac{\mu_3(S)}{s} + \frac{\mu_2(S)}{\sqrt{s}} t/3} \right).
\end{align*}
Taking the union bound over $i\in S^c$ completes the proof.
\end{proof}

\begin{lemme}
\label{lem:offSupportCoherence}
Let $S$ be a subset of $\{1, \hdots , n \}$. Then, for any $0<t < \frac{\mu_1(S)}{\mu_2(S)}$, one has that
\begin{align}
\tag{E4}
\label{lem:E4eq}
\Pbb & \left( \max_{i\in {S^c}} \left\| \Ab_S^* \Ab \eb_{i} \right\|_{2} \geq   t \right)  \leq n \exp \left( - \frac{\left( \sqrt{m / \mu_1(S)} t - 1\right)^2}{4} \right).
\end{align}
\end{lemme}

\begin{proof}
Let us fix some $i\in  {S^c}$. For $k=1, \hdots, m$, we define the random matrix $$\xb_k := \frac{1}{m} \Bb_{k,S}^{*} \Bb_{k} \eb_{i} .$$  One has that $\Ebb \xb_k = 0$. Then, we remark that
\begin{align*}
\left\| \Ab_S^* \Ab \eb_{i} \right\|_{2}  = \left\| \frac{1}{m} \sum_{k=1}^m \Bb_{k,S}^{*} \Bb_{k} \eb_{i} \right\|_{2 } =\left\| \sum_{k=1}^m \xb_k \right\|_{2}.
\end{align*}
It follows that
\begin{align*}
\|\xb_k\|_{2 }  =\frac{1}{m } \left\| \Bb_{k,S}^{*}  \Bb_{k} \eb_{i} \right\|_{2}  \leq \frac{1}{m } \frac{ \mu_2(S)}{\sqrt{s}}.
\end{align*}
Furthermore, using Cauchy-Schwarz inequality, one has that
\begin{align*}
\Ebb \left\| \xb_1 \right\|_2^2 &= \frac{1}{m^2} \Ebb \|\Bb_{1,S}^{*}  \Bb_1 \eb_i \|_2^2 \leq  \frac{1}{m^2} \Ebb \| \Bb_{1,S}^{*}  \|_{2\rightarrow 2}^2 \| \Bb_1 \eb_i \|_2^2 \leq \frac{1}{m^2}  \mu_1(S) \Ebb \| \Bb_1 \eb_i \|_2^2 = \frac{1}{m^2}  \mu_1(S) \|  \eb_i \|_2 \\
&\leq \frac{1}{m^2}  \mu_1(S).
\end{align*}
Hence, using the above upper bounds, it follows from Bernstein's inequality for random vectors (see Theorem \ref{theo:VectBern2}) that
\begin{align*}
\Pbb \left( \left\| \Ab_S^* \Ab \eb_{i} \right\|_{2}  \geq  t\right) \leq  \exp \left( - \frac{\left( \sqrt{m / \mu_1(S)} t - 1\right)^2}{4} \right),
\end{align*}
Finally, Inequality \eqref{lem:E4eq} follows from a union bound over  $i\in {S^c}$, which completes the proof.
\end{proof}

\section{Proofs of the main results}

\subsection{Proof of Theorem \ref{thm:recovery}}
\label{sec:proof1}
In this section, we recall an inexact duality formulation of the minimization problem \eqref{pb:min} in the form of sufficient conditions to guarantee that the vector $\xb$ is the unique minimizer of  \eqref{pb:min}, see \cite{candes2011probabilistic}. These conditions give  the properties that an inexact dual vector must satisfy to ensure the uniqueness of the solution of \eqref{pb:min}. In what follows, the notation $\Mb_{\mid R}$ denotes the restriction of a square matrix $\Mb$ to its range $R$, and we define
$$
\|   \Mb_{\mid R}^{-1} \|_{2 \rightarrow 2} = \sup_{ \xb \in R ;  \;  \| \xb \|_{2} = 1 } \|  \Mb_{\mid R}^{-1} \xb \|_{2}
$$
as the operator norm of the inverse of $\Mb_{\mid R}$ restricted to its range.
\begin{lemme}[Inexact duality \cite{candes2011probabilistic}]
\label{lem:inexactDuality}
Suppose that $\xb \in \Rbb^{n}$ is supported on $S \subset \{1, \hdots , n\}$.  Then, assume that
\begin{equation}
\| \left(\Ab^*_S \Ab_S \right)_{\mid S}^{-1} \|_{2 \rightarrow 2} \leq 2 \qquad \text{and} \qquad  \max_{i\in {S^c}} \left\|  \Ab_S^* \Ab \eb_i \right\|_{2} \leq 1. \label{ass:dual1}
\end{equation}
Morever, suppose that there exists $\vb \in \Rbb^n$ in the row space of $\Ab$ obeying
\begin{equation}
\| \vb_S - \sgn(\xb_S)\|_2 \leq 1/4 \qquad \text{and} \qquad  \|  \vb_{S^c}  \|_\infty\leq 1/4, \label{ass:dual2}
\end{equation}
Then, the vector $\xb$ is the unique solution of the minimization problem \eqref{pb:min}
\end{lemme}

First, let us focus on Conditions \eqref{ass:dual1}.  We can remark that
$$ \| \left(\Ab^*_S \Ab_S \right)_{\mid S}^{-1} \|_{2\rightarrow 2}  = \left\| \sum_{k=1}^\infty \left( \Ab_S^* \Ab_S -\Id_s \right)^k \right\|_{2\rightarrow 2}  \leq \sum_{k=1}^\infty  \left\|  \Ab_S^* \Ab_S - \Id_s  \right\|_{2\rightarrow 2} ^k.
$$
Therefore, if the condition  $\left\| \Ab_S^* \Ab_S - \Id_s \right\|_{2\rightarrow 2}  \leq 1/2$  is satisfied, then   $\| \left(\Ab^*_S \Ab_S \right)_{\mid S}^{-1} \|_{2\rightarrow 2}  \leq 2$. Hence, by Lemma \ref{lem:localIsometry}, it is clear that $\| \left(\Ab^*_S \Ab_S \right)_{\mid S}^{-1} \|_{2\rightarrow 2}  \leq 2$  with probability at least $1-\varepsilon$, provided that
$$
m \geq 8 \left( \mu_1(S) + \frac{1}{6} \max \left( \mu_1(S) -1 , 1 \right) \right) \log \left( \frac{2 s  }{\varepsilon}  \right) .
$$
By definition of $\gamma(S)$, the first inequality of Conditions \eqref{ass:dual1} is ensured with probability larger than $1-\varepsilon$ if
\begin{equation}
\label{eq:condm2}
m \geq 8 \left(  \gamma(S) + \frac{1}{6} \max \left( \gamma(S)-1 , 1 \right) \right) \log \left( \frac{2 s  }{\varepsilon}  \right).
\end{equation}
Furthermore, using Lemma \ref{lem:offSupportCoherence}, we obtain that 
$$
\max_{i\in {S^c}} \| \Ab_S^* \Ab \eb_{i} \|_{2}  \leq 1
$$
 with probability larger than $1-\varepsilon$ if
$$
m \geq \mu_1(S) \left( 1+ 4 \sqrt{\log \left( \frac{n}{\varepsilon} \right)} + 4 {\log \left( \frac{n}{\varepsilon} \right)} \right).
$$
Again by definition of $\gamma(S)$, the second part of Conditions \eqref{ass:dual2} is ensured if
\begin{equation}
\label{eq:condm3}
m \geq 9\gamma(S) \log \left( \frac{n }{\varepsilon}\right). 
\end{equation}

Conditions \eqref{ass:dual2} remain to be verified. The rest proof of Theorem \ref{thm:recovery} relies on the construction of a vector $\vb$ satisfying the conditions described in Lemma \ref{lem:inexactDuality} with high probability.To do so, we adapt the so-called golfing scheme introduced by Gross \cite{gross2011recovering} to our setting. More precisely, we will iteratively construct a vector that converges to a vector $\vb$ satisfying \eqref{ass:dual2} with high probability.

Let us first partition the sensing matrix $\Ab$ into blocks of blocks so that, from now on, we denote by $\Ab^{(1)}$  the first $m_1$ blocks of $\Ab$, $\Ab^{(2)}$ the next $m_2$ blocks, and so on. The $L$ random matrices $\left\{ \Ab^{(\ell)} \right\}_{\ell=1,\ldots,L}$ are independently distributed, and we have that $m = m_1+m_2+\hdots+m_L$. As explained before, $\Ab^{(\ell)}_S$ denotes the matrix $\Ab^{(\ell)} P_{S}$.
The golfing scheme starts by defining $\vb^{(0)}=0$, and then it inductively defines
\begin{align}
\vb^{(\ell)} = \frac{m}{m_\ell} \Ab^{(\ell)^*} \Ab^{(\ell)}_S \left( \eb - \vb_{S}^{(\ell-1)} \right) +\vb^{(\ell-1)},
\end{align}
for $\ell=1,\hdots, L$. In the rest of the proof, we set $\vb = \vb^{(L)}$. By construction, $\vb$ is in the row space of $\Ab$. The main idea of the golfing scheme is then to combine the results from the various Lemmas in Section \ref{sec:aux} with an appropriate choice of $L$ and the number $m$ of measurements, to show that the random vector $\vb$ will satisfy the assumptions of Lemma \ref{lem:inexactDuality} with large probability.
Using the shorthand notation $\vb_{S}^{(\ell)} = P_{S}^* \vb^{(\ell)}$, let us define
$$
\wb^{(\ell)} =  \eb -\vb_{S}^{(\ell)}, \; \ell = 1,\ldots,L,
$$
where $\eb = \sgn (\xb_S)$, and  $\xb \in \Rbb^{n}$ is an $s$-sparse vector supported on S. 

From the definition of $\vb_{S}^{(\ell)}$, it follows that, for any $1 \leq \ell \leq L$,
\begin{equation}
\label{eq:qirec}
\wb^{(\ell)} = \left( \Id_s -\frac{m}{m_\ell}  \Ab^{(\ell)*}_{S} \Ab^{(\ell)}_{S} \right) \wb^{(\ell-1)} = \prod_{j=1}^\ell \left( \Id_s -\frac{m}{m_j}  \Ab^{(j)*}_{S} \Ab^{(j)}_{S} \right)   \eb,
\end{equation}
and
\begin{equation}
\label{eq:vfctq}
\vb = \sum_{\ell=1}^L \frac{m}{m_\ell} \Ab^{(\ell)*} \Ab^{(\ell)}_{S} \wb^{(\ell-1)}.
\end{equation}
Note that in particular, $\wb^{(0)} =\eb$ and $\wb^{(L)} =\eb -  \vb_{S}$.
In what follows, it will be shown  that the matrices $\Id_s - \frac{m}{m_\ell} \Ab^{(\ell)*}_S \Ab^{(\ell)}_{S}$ are contractions, and that the norm of the  vector $\wb^{(\ell)}$ decreases geometrically fast as $\ell$ increases. Therefore, $\vb^{(\ell)}_{S}$ becomes close to $\eb$ as $\ell$ tends to $L$. In particular, we will prove that $\|\wb^{(L)} \|_2  \leq 1/4$ for a suitable choice of $L$. In addition, we also show that $\vb$ satisfies the condition $\| \vb_{S^c} \|_{\infty}  \leq 1/4$. All these conditions will be shown to be satisfied with a large probability (depending on $\varepsilon$).

For all $1 \leq \ell \leq L$, we assume that with high probability

\begin{align}
\label{eq:control2}
\left\| \wb^{(\ell)} \right\|_2 &\leq \underbrace{\left( \sqrt{\frac{\mu_1(S)-1}{m_\ell}} + r_\ell  \right)}_{r_\ell'}\left\| \wb^{(\ell-1)} \right\|_2 \\
\label{eq:control3}
\left\| \frac{m}{m_\ell} \left(\Ab_{S^c}^{(\ell)}\right)^* \Ab^{(\ell)}_{S} \wb^{(\ell-1)}\right\|_\infty &\leq t_\ell \| \wb^{(\ell-1)} \|_2. 
\end{align}
The values of the quantities $t_\ell$ and $r_\ell$, introduced in the above equations, will be specified later in the proof.
Note that using \eqref{eq:control2}, we can write that
\begin{align}
\label{eq:eminusv}
\left\| \eb - \vb_S \right\|_2 = \| \wb^{(L)} \|_2 \leq \| \eb \|_2 \prod_{\ell=1}^L r_\ell' \leq \sqrt{s} \prod_{\ell=1}^L r_\ell'.
\end{align}
Furthermore, Equation \eqref{eq:control3} implies that
\begin{align}
 \| \vb_{S^c} \|_\infty &= \left\| \sum_{\ell=1}^L \frac{m}{m_\ell} \left(\Ab_{S^c}^{(\ell)}\right)^* \Ab^{(\ell)}_{S} \wb^{(\ell-1)} \right\|_\infty \nonumber \\
&\leq \sum_{\ell=1}^L  \left\|  \frac{m}{m_\ell} \left(\Ab_{S^c}^{(\ell)}\right)^*\Ab^{(\ell)}_{S} \wb^{(\ell-1)} \right\|_\infty \nonumber  \\
&\leq \sum_{\ell=1}^L t_\ell \left\|   \wb^{(\ell-1)} \right\|_\infty
\nonumber \\ 
&\leq \sqrt{s}\sum_{\ell=1}^L t_\ell \prod_{j=1}^{\ell-1} r_j'. \label{eq:tell}
\end{align}

We denote by $p_1(\ell)$ and $p_2(\ell)$ the probability that the upper bound  \eqref{eq:control2}, \eqref{eq:control3} do not hold. Now, let us set the number of blocks of blocks $L$, the number of blocks $m_\ell$ in each $\Ab^{(\ell)}$ and the values of the parameters $t_\ell$ and $r_\ell$ that have been introduced above.
We propose to make the following choices : 
\begin{enumerate}
\item $L = 2 + \left\lceil   \frac{\log \left(s \right)}{ 2 \log 2}   \right\rceil  $,
 \item  $m_1, m_2 \geq c \gamma(S) \log \left( 4n \right) \log\left (2 \varepsilon^{-1}\right)$\\
  $m_\ell \geq c \gamma(S) \log \left( 2L \varepsilon^{-1}\right),$ for $\ell=3,\hdots, L$, for some sufficiently large $c \geq 1$,
\item $r_1, r_2 = \frac{1}{4\sqrt{\log 4 n}}$, \\
$r_\ell = \frac{1}{4},$ for $\ell=3,\hdots, L$,
\item $t_1, t_2 = \frac{1}{8\sqrt{s}}$, \\
$t_\ell = \frac{\log(4n)}{8\sqrt{s}},$ for $\ell=3,\hdots, L$.
\end{enumerate}
With such choices, we obtain that
$$ r_1' , r_2'= \sqrt{ \frac{\mu_1(S) -1}{m_\ell}} + \frac{1}{4\sqrt{\log n}} \leq \frac{1}{2\sqrt{\log n}} \leq \frac{1}{2},
$$
and
$$ r_\ell' =\sqrt{ \frac{\mu_1(S) -1}{m_\ell}} + \frac{1}{4} \leq \frac{1}{2}
$$
Furthermore, using \eqref{eq:eminusv}, we obtain that
\begin{equation}
 \left\| \eb-\vb_S \right\|_2 \leq  \sqrt{s}  \prod_{\ell=1}^L r_\ell' \leq \frac{  \sqrt{s}}{2^{L}} \leq \frac{1}{4}, \label{eq:conrand1}
\end{equation}
where the last inequality follows from the  previously specified choice on $L$. Moreover, using \eqref{eq:tell}, we have that
\begin{align}
\notag
\| \vb_{S^c}\|_\infty &\leq  \sqrt{s}\sum_{\ell=1}^L t_\ell \prod_{j=1}^{\ell-1} r_j' =\sqrt{s}\left(  t_1 + t_2 r_1' + t_3 r_1' r_2' +... \right)\\
\notag
&\leq \left( \frac{1}{8} + \frac{1}{16 \sqrt{\log n }} + \frac{1}{32} + ... \right) \\
&\leq \frac{1}{4}.
 \label{eq:conrand2}
\end{align}
For such a choice of parameters, and by Lemmas \ref{lem:lowDistortion} and \ref{lem:distortionInf2}, if we fix $\varepsilon \in (0, 1/6)$, the bound $c\geq 534$ ensures $p_1(1), p_1(2)$, $p_2(1)$, $p_2(2) \leq \varepsilon/2$ and $p_1(\ell), p_2(\ell)  \leq \varepsilon/2L$ for $\ell=3,\hdots, L$.
Therefore, $\sum_{\ell=1}^L p_1(\ell) \leq 2\varepsilon$ and $\sum_{\ell=1}^L p_2(\ell) \leq 2\varepsilon$. 
%It can be checked that our choices on the parameters $t_\ell$ and $r_\ell$, imply that $m_{\ell}$ will satisfy Inequalities \eqref{eq:cond1m}, \eqref{eq:cond2m}, and \eqref{eq:cond3m} provided that the constant $c$ is sufficiently large. We can bound the probability of failure of Inequailities \eqref{eq:conrand1} and \eqref{eq:conrand1} by
%$$ 
%\sum_{\ell=1}^L p_0(\ell) + p_1(\ell) +p_2(\ell) \leq 4 \varepsilon.
%$$
From the above calculation, and by Lemmas \ref{lem:lowDistortion} and \ref{lem:distortionInf2} we finally obtain that if the overall number $m$ of blocks samples obeys the condition
\begin{equation*}
m = \sum_{\ell=1}^L m_\ell \geq c {\gamma(S)} \left( 2 \log \left( 4n \right) \log\left (2 \varepsilon^{-1}\right) + (L-2)\log \left( 2L \varepsilon^{-1}\right)\right), 
\end{equation*}
which can be simplified into
\begin{equation}
m \geq c {\gamma(S)} \left( 2 \log \left( 4n \right) \log\left (2 \varepsilon^{-1}\right) + \log s \log \left( 2e \log(s) \varepsilon^{-1}\right)\right), 
\label{eq:condm1}
\end{equation}
 then the random vector $\vb$, defined by \eqref{eq:vfctq}, satisfies Assumptions \ref{ass:dual2} of Lemma \ref{lem:inexactDuality} with probability larger than $1-4\varepsilon$.

Hence, we have thus shown that if $m$ satisfies the conditions  \eqref{eq:condm2}, \eqref{eq:condm3} and \eqref{eq:condm1}, then the Assumptions  \ref{ass:dual1} and  \ref{ass:dual2} of Lemma \ref{lem:inexactDuality} simultaneously hold  with probability larger than $1-6\varepsilon$. Note that the bound \eqref{eq:condm1} is stronger than \eqref{eq:condm2} and \eqref{eq:condm3}. We complete the proof of Theorem \ref{thm:recovery} by replacing $\varepsilon$ by $\varepsilon /6$.
The final result on the required number of blocks measurements reads as follows
$$  m \geq  c {\gamma(S)} \left( 2 \log \left( 4n \right) \log\left (12 \varepsilon^{-1}\right) + \log s \log \left( 12e \log(s) \varepsilon^{-1}\right)\right),
$$
for $c= 534$, but in the statement we simplify the expression to improve the readability. Moreover, note that in our proof, for the sake of concision, there is no attempt to strenghten the previous result. Yet, we could have used the clever trick used in \cite{candes2011probabilistic}, and reused in \cite{foucart2013mathematical} which consists of oversampling blocks in the golfing scheme.

\subsection{Proof of Proposition \ref{prop:gammaSs}}
\label{proof:gamma}

Since $\gamma(S)=\max(\mu_1,\mu_2,\mu_3)$, it suffices to show that setting $\mu_i = s\mu_4$ for $i\in \{1,2,3\}$ is sufficient to ensure the inequalities \eqref{ineq:conditions}.

The first inequality in \eqref{ineq:conditions} can be shown as follows:
$$\left\| {\Bb_S}^*  {\Bb_S}  \right\|_{2\rightarrow 2} \leq  {\left\| {\Bb_S}^*  {\Bb_S}  \right\|_{\infty \rightarrow \infty} } \leq s {\left\| {\Bb}^*  {\Bb}  \right\|_{1\rightarrow \infty} }\leq s\mu_4.$$

The second inequality in \eqref{ineq:conditions} can be shown as follows:
\begin{align*}
 \sqrt{s} \max_{i \in S^c} {\left\| {\Bb_S}^* {\Bb} \eb_i \right\|_2}\leq \sqrt{s} \sqrt{s } {\left\| {\Bb}^* {\Bb}  \right\|_{1\rightarrow \infty}}\leq s\mu_4.
\end{align*} 

Finally, fix $i\in S^c$.  One can write
\begin{align*}
s \Ebb {\Bb_S}^* \left( {\Bb} \eb_i \right) \left( {\Bb} \eb_i \right)^* {\Bb_S} &\preceq s \| \left( {\Bb} \eb_i \right) \left( {\Bb} \eb_i \right)^* \|_{2\rightarrow 2} \Ebb {\Bb_{S}}^* \Bb_{S} \\ 
&\preceq  s \max_i \| {\Bb} \eb_i  \|_2^2 \Id \\
& \preceq  s  \left\| {\Bb}^* {\Bb}  \right\|_{1\rightarrow \infty} \Id \\
&\preceq  s\mu_4 \Id.
\end{align*}

\subsection{Proof of Proposition \ref{prop:gaussian}}

Let us evaluate the quantities $\left( \mu_i(S) \right)_{1\leq i \leq 3}$ introduced in Definition \ref{def:quantities} to upper bound $\gamma(S)$ with high probability.
For this purpose, using Theorem 2 in \cite{ledoux2010small}, we get that for any $0<t<1$
\begin{align}
\label{ineq:ledoux}
 \Pbb \left( \| \Bb_{S}^* \Bb_{S} \|_{2\rightarrow 2} \geq \left( 1 + \sqrt{\frac{s}{p}} \right)^2 (1+t)  \right) \leq C \exp\left( -\sqrt{ps} t^{3/2} \left( \frac{1}{\sqrt{t}}   \wedge      \left( \frac{s}{p} \right)^{1/4}  \right) /C \right),
\end{align} 
for $C$ a universal constant, under the assumption that $s>p$. We could also treat the case where $p>s$ by inverting the role of $s$ and $p$ in the above deviation inequality. We restrict our study to the case $s>p$ for simplicity.

By Inequality \eqref{ineq:ledoux}, we can consider that $\mu_1(S) \lesssim \frac{s}{p}$ with large probability (provided that $s$ is sufficiently large).
For evaluating $\mu_2(S)$, we use the following upper bound,
\begin{align*}
\max_{i\in S^c} \| \Bb_{S}^* \Bb \eb_i \|_{2} \leq \max_{i\in S^c} \| \Bb_{S}^* \|_{2\rightarrow 2} \| \Bb  \eb_i \|_{2} \leq \sqrt{\| \Bb_{S}^* \Bb_{S} \|_{2\rightarrow 2}} \max_{i\in S^c} \sqrt{\| \Bb  \eb_i \|_{2}^2 }.
\end{align*}
We already know that the first term  $\sqrt{\| \Bb_{S}^* \Bb_{S} \|_{2\rightarrow 2}}$ in the above inequality is bounded by $\sqrt{\frac{s}{p}}$ (up to a constant)  with high probability, thanks to the previous discussion on $\mu_1(S)$. 
As for the second term, we use a union bound and the sub-gamma property of the chi-squared distribution, see \cite[p.29]{boucheron2013concentration}, to derive that
\begin{align*}
\Pbb \left( \max_{i\in S^c} \| \Bb \eb_i \|_2^2 \geq 2\left( \sqrt{\frac{t}{p}} + \frac{t}{p}\right)  \right) \leq (n-s)  \exp(-t)\leq n  \exp(-t) .
\end{align*}
Let $\delta >1$. Using the above deviation inequality, we get that 
$$ \max_{i\in S^c} \sqrt{\| \Bb  \eb_i \|_{2}^2 } \lesssim \sqrt{  \frac{\delta\log (s) }{p} },
$$
with probability larger than $1-ns^{-\delta}$.
Thus, we get the following upper bound for $\mu_2(S)$:
\begin{align*}
\mu_2(S) \lesssim  \frac{s\sqrt{\delta \log(s)}}{p},
\end{align*}
that holds with high probability provided that $s$ is sufficiently large.
Finally, by conditioning with respect to $\Bb_S$ and using the independence of $\Bb_S$ and $\Bb\eb_i$ for $i \in S^c$, we have that
\begin{align*}
s \max_{i\in S^c} &\left\| \Ebb \left( \Bb_{S}^* \left( \Bb\eb_i \right) \left( \Bb \eb_i \right)^*  \Bb_{S} \right) \right\|_{2\rightarrow 2} = s  \max_{i\in S^c} \left\| \Ebb \left[ \Ebb \left( \Bb_{S}^* \left( \Bb \eb_i \right) \left( \Bb \eb_i \right)^*  \Bb_{S}  | \Bb_{S}\right) \right] \right\|_{2\rightarrow 2} , \\
%&= s \max_{i\in S^c} \left\| \Ebb\left[ \Bb_{S}^* \Ebb \left(  \left( \Bb \eb_i \right) \left( \Bb \eb_i \right)^*   | \Bb_{S} \right) \Bb_{S}  \right] \right\|_{2\rightarrow 2} , \\
&= s \max_{i\in S^c} \left\| \Ebb\left[ \Bb_{S}^* \Ebb \left(  \left( \Bb \eb_i \right) \left( \Bb \eb_i \right)^* \right) \Bb_{S} \right]  \right\|_{2\rightarrow 2} = s \max_{i\in S^c} \left\| \Ebb\left[ \Bb_{S}^* \frac{1}{p} \Id \Bb_{S}  \right]  \right\|_{2\rightarrow 2}  = \frac{s}{p}.
\end{align*}
Hence, one can take $\mu_3(S) = \frac{s}{p}$.
Combining all these estimates we get that $\gamma(S)\lesssim \frac{s}{p} \sqrt{\delta \log(s)}$. Therefore,  assuming that the lower bound on $m$ in Theorem \ref{thm:recovery} still holds in the case of acquisition by blocks made of Gaussian entries, we need $m = O\left(\frac{s}{p} \log(s) \log(n)\right)$ blocks of measurements to ensure exact recovery, that is an overall number of measurements $q=O(s\log (s)\log(n))$.

\subsection{Proof of Proposition \ref{prop:tensorlimit}}

The proof is divided in two parts. First we show the result for $1\leq s\leq \sqrt{n}$ and then we show it for $\sqrt{n}<s\leq n$.
We let $\eb_i$ denote the $i$-th element of the canonical basis.

\textbf{Part 1:}
Fix $s \in \{1,\hdots,\sqrt{n}\}$.
Let $\Cc_s$ denote the class of vectors of kind $\xb=\alpha\otimes \eb_1$, where $\alphab \in \Rbb^{\sqrt{n}}$ is $s$-sparse. 
Note that every $\xb\in \Cc_s$ is $s$-sparse and that
\begin{align*}
 \Ab x&= (\widetilde{\Psib}_{K,:} \otimes \Psib)\cdot (\alphab \otimes \eb_1) \\
&= \left(\widetilde{\Psib}_{K,:} \alpha \right) \otimes \Psib \eb_1.
\end{align*}

 In order to identify every $s$-sparse $\xb$ knowing $\yb=\Ab \xb$, there should not exist two distinct $s$-sparse vectors $\alphab^{(1)}$ and $\alphab^{(2)}$ in $\Cbb^{\sqrt{n}}$ such that $\widetilde{\Psib}_{K,:} \alphab^{(1)}=\widetilde{\Psib}_{K,:} \alphab^{(2)}$. The vector $\alphab^{(1)}-\alphab^{(2)}$ is $\min(2s,\sqrt{n})$-sparse. Therefore, a necessary condition for recovering all $s$-sparse vectors with $1\leq s \leq \sqrt{n}$ is that $\widetilde{\Psib}_{K,:} \alphab \neq 0$ for all non-zero $\min(2s,\sqrt{n})$-sparse vectors $\alpha$. To finish the first part of the proof it suffices to remark that a necessary condition for a set of $\min(2s,\sqrt{n})$ columns of $\widetilde{\Psib}_{K,:}$ to be linearly independent is that $m=|K|\geq \min(2s,\sqrt{n})$, see Lemma \ref{lem:cohen}. 

\textbf{Part 2:}
Assume that $\sqrt{n} < s\leq n$. Consider the class $\mathcal{C}_s$ of $s$-sparse vectors of kind $\displaystyle \xb=\sum_{l=1}^{\sqrt{n}} \alphab^{(l)}\otimes \eb_l$, where $\supp(\alphab^{(1)})=\{1,\hdots,\sqrt{n}\}$. For $\xb\in \mathcal{C}_s$
\begin{equation*}
 \Ab \xb =\sum_{l=1}^{\sqrt{n}} \left( \widetilde{\Psib}_{K,:} \alphab^{(l)} \right) \otimes \Psib \eb_l.
\end{equation*}
Similarly to the first part of the proof, in order to identify every $s$-sparse vectors, there should not exist $\alphab^{(1)}$ and $\alphab^{(1)'}$ with support equal to $\{1,\hdots,\sqrt{n}\}$ such that $\widetilde{\Psib}_{K,:} \alphab^{(1)}=\widetilde{\Psib}_{K,:} \alphab^{(1)'}$. We showed in the previous section that a necessary condition for this condition to hold is $m=\sqrt{n}$.

\subsection{Proof of Proposition \ref{prop:2DFour2}}
\label{comput2DFourier}

We consider blocks that consist of discrete lines in the 2D Fourier space as in Fig~\ref{fig:schema}(b). 
We assume that $\sqrt{n} \in \Nbb$ and that $\Ab_0$ is the 2D Fourier matrix applicable on $\sqrt{n}\times \sqrt{n}$ images. 
For all $p_1 \in \left\lbrace 1, \hdots , \sqrt{n} \right\rbrace$,
\begin{align}
 \Bb_{p_{1}}  &= \left[  \frac{1}{\sqrt{n}} \exp \left( 2i\pi \left( \frac{p_1 \ell_1 + p_2 \ell_2}{\sqrt{n}} \right) \right) \right]_{ \displaystyle (p_1 , p_2)(\ell_1 ,\ell_2)      } 
\end{align}
with $1 \leq p_2 \leq \sqrt{n}  ,  1 \leq \ell_1 ,\ell_2 \leq \sqrt{n}$.
Let
$ S \subset \left\lbrace 1 , \hdots , \sqrt{n} \right\rbrace \times \left\lbrace 1 , \hdots , \sqrt{n} \right\rbrace
$ denote the support of $\xb$,
with $|{S}| =s$.
By definition of the 2D Fourier matrix of size $n \times n$, $\left\| \Bb_k^* \Bb_k \right\|_{1\rightarrow \infty} = 1/\sqrt{n}$, for all $k \in \left\lbrace 1 , \hdots , \sqrt{n} \right\rbrace$. Thus, Theorem \ref{thm:recovery} leads to
$$
m \geq c s  \frac{1}{\sqrt{n}} \max_{1 \leq k \leq M} \frac{1}{\pi_k}   \log \left( 4n \right) \log\left (12 \varepsilon^{-1}\right)  .
$$
Therefore, the choice of an optimal drawing probability, regarding the number of measurements, is given by 
$$ \pi_k^\star = \frac{1}{\sqrt{n}},  \quad \forall k \in  \left\lbrace 1,\hdots,\sqrt{n} \right\rbrace
$$
and the number of measurements can be written as follows
$$
m \geq C  s  \log \left( 4n \right) \log\left (12 \varepsilon^{-1}\right),
$$
which ends the proof of Proposition \ref{prop:2DFour2}.

\section{An example with overlapping blocks}
\label{app:overlap}

Let us illustrate the overlapping setting, in the case of blocks that consist in rows and columns in the 2D Fourier domain. Matrix $\Ab_0 \in \Cbb^{n \times n}$ is the 2D Fourier transform matrix. We set 
$$ I_k^{\text{row}} = \left\lbrace i \in \left\lbrace 1, \hdots , n \right\rbrace, \;  (k-1)\sqrt{n} \leq i \leq k \sqrt{n} \right\rbrace
\qquad 
I_k^{\text{col}} = \left\lbrace k, \sqrt{n} + k, \hdots , (\sqrt{n}-1)\sqrt{n} + k\right\rbrace
$$
the sets of indexes  of $\left( \ab_i^*\right)_{i \in \left\lbrace 1, \hdots , n \right\rbrace} $ that respectively correspond to the $k$-th row and the $k$-column in the 2D Fourier plane.
Then, we can write the blocks as follows:
\[
\Bb_k = 
\left\{
\begin{array}{ll}
\left( \frac{1}{\sqrt{2}} \ab_i^* \right)_{i \in I_k^{\text{row}}} &  \text{if} \; k \in \left\lbrace 1 , \hdots , \sqrt{n} \right\rbrace \\
\left( \frac{1}{\sqrt{2}} \ab_i^* \right)_{i \in I_{k-\sqrt{n}}^{\text{col}}} &  \text{if} \; k \in \left\lbrace \sqrt{n}+1  , \hdots , 2\sqrt{n} \right\rbrace.
\end{array}
\right.
\]
We have chosen the normalization factor equal to $1/\sqrt{2}$, as suggested, since each pixel of the image belongs to two blocks: one row and one column. 
According to Corollary \ref{cor:optimalpi}, we conclude that the required number of blocks of measurements must satisfy
\begin{align}
m &\geq c s \frac{1}{2\sqrt{n}} \max_{1 \leq k \leq M} \frac{1 }{\pi_k} \left( 2 \log \left( 4n \right) \log\left (12 \varepsilon^{-1}\right) + \log s \log \left( 12e \log(s) \varepsilon^{-1}\right)\right).
\end{align}
Choosing the uniform probability for $\Pi^\star$, i.e. $\pi^\star_k=\frac{1}{2\sqrt{n}}$ for all $k  \in \left\lbrace 1, \hdots , 2\sqrt{n} \right\rbrace$ leads to the following number of blocks of measurements
\begin{align}
m &\geq c s  \left( 2 \log \left( 4n \right) \log\left (12 \varepsilon^{-1}\right) + \log s \log \left( 12e \log(s) \varepsilon^{-1}\right)\right),
\end{align}
which is the same requirement in the 2D Fourier domain without overlapping, see Proposition \ref{prop:2DFour2}. 
%The same argument can be applied to the case of the 2D Fourier-Shannon's wavelets covered by Proposition \ref{prop:2DFW1}.

\bibliographystyle{alpha}
\bibliography{mybib}

\end{document}